\documentclass[11pt]{article}

\usepackage{graphicx}%
\usepackage{multirow}%
\usepackage{amsmath,amssymb,amsfonts}%
\usepackage{amsthm}%
\usepackage{mathrsfs}%
\usepackage[title]{appendix}%
\usepackage{xcolor}%
\usepackage{textcomp}%
\usepackage{manyfoot}%
\usepackage{booktabs}%
\usepackage{algorithm}%
\usepackage{algorithmicx}%
\usepackage{algpseudocode}%
\usepackage{listings}%
\usepackage{tabularx}
\usepackage{comment}
\usepackage{tikz}
\usetikzlibrary{positioning}
\usetikzlibrary{shapes.geometric, arrows.meta}
\usepackage{rotating}
\usepackage{float}
\usepackage{stfloats} 
\usepackage{placeins}

\usepackage{caption}
\usepackage{pdflscape}

\usepackage[margin=1in]{geometry}
\usepackage[numbers,sort&compress]{natbib}
\usepackage{hyperref}

\theoremstyle{plain}%
\newtheorem{theorem}{Theorem}%
\newtheorem{lemma}{Lemma}
\newtheorem{corollary}{Corollary}

\theoremstyle{remark}%
\newtheorem{remark}{Remark}%

\theoremstyle{definition}%
\newtheorem{definition}{Definition}%
\newtheorem{property}{Property}

\begin{document}

\title{Generalised Bit-Vector Abstractions for Formal Verification of Quantum Error-Detection and Entanglement Circuits over \texorpdfstring{$\{H,X,\text{C-NOT}\}$}{\{H, X, C-NOT\}}: CSS Constructions, Soundness, and Mutation-Based Validation}

\author{Arun Govindankutty}

\date{}

\maketitle

\begin{abstract}
Formal verification of quantum circuits becomes increasingly difficult as circuit size increases,
since exhaustive reasoning over the underlying Hilbert space is computationally expensive.
In previous work~\cite{bib34}, we introduced a bit-vector abstraction that separates the tracking of superposition from basis-state behaviour,
enabling SMT-based verification of fixed-size error-detection and entanglement circuits constructed from
$H$, $X$, and C-NOT gates. This paper extends that framework in five directions.
First, we generalise the correctness properties to parametrised repetition and
concatenated phase-flip codes with up to 101 logical elements.
Second, we verify syndrome extraction for CSS codes specified by arbitrary parity-check matrices,
including the Steane code, the $[[15,7,3]]$ quantum Hamming code, the $[[15,1,3]]$ quantum Reed-Muller code,
and rotated surface codes with distances up to 101.
The verification considers arbitrary error patterns, error detection up to a specified weight,
and recovery using a classical decoder. Counterexamples are obtained for errors beyond the correction or detection capability of the respective codes.
Third, we characterise the fragment of $\{H,X,\text{C-NOT}\}$ circuits for which the abstraction faithfully represents the corresponding Hilbert-space semantics.
We prove soundness for this fragment, demonstrate the necessity of the stated conditions through counterexamples,
validate the characterisation by differential testing against state-vector simulation,
and prove the six lemmas established in the original work.
Fourth, we conduct a mutation study using 13 mutation operators and automatically identify equivalent mutants.
Within the evaluated fault model, the approach detects all 622 non-equivalent mutants.
Fifth, we implement the approach directly using the Z3 API and verify GHZ circuits with up to 262,144 qubits across three topologies on a 16~GB laptop.
The evaluation includes repeated measurements and a controlled ablation of the encoding pipeline and solver version.
Experiments with random $H$+C-NOT networks show that verification cost increases steeply with circuit depth when outputs depend on a large number of inputs.
The approach is not intended to cover general stabiliser circuits, $Y$-type stabilisers, fault-tolerant protocols, or physical noise. We identify the extensions required to address these cases.
\end{abstract}


\section{Introduction}
\label{sec:intro}

Quantum computing has the potential to outperform classical computing in several application domains, including communication, optimization, machine learning, and materials science~\cite{bib1,bib2,bib3,bib4,bib5}. 
Achieving this potential depends on two fundamental capabilities: the generation and preservation of quantum entanglement across multiple qubits, and the reliable detection and correction of errors arising from noise and decoherence in quantum hardware~\cite{bib6,bib7,bib8,bib9}. 
These capabilities are realized through quantum circuits whose correctness is essential for the reliability of larger quantum algorithms and protocols.
Even a small implementation error, such as an omitted controlled-NOT (C-NOT) gate, an incorrectly placed Hadamard gate, or an improperly connected ancilla qubit, may alter the intended circuit behaviour. 
Such faults can invalidate entanglement generation or compromise error-detection mechanisms, thereby affecting the correctness of subsequent quantum computations.

Verifying the functional correctness of quantum circuits remains challenging because quantum computation differs fundamentally from classical computation.
The state space of an $n$-qubit system grows exponentially with the number of qubits, measurement outcomes are inherently probabilistic and disturb the quantum state, and arbitrary quantum states generally do not admit compact classical representations.
Formal verification techniques provide a systematic alternative to simulation by reasoning about all possible behaviours of a circuit rather than individual executions~\cite{bib10,bib11,bib12}.
Despite this advantage, the scalability of existing verification methods remains limited.
Rewrite-rule-based verification using dyadic arithmetic has been demonstrated for circuits of up to 96 qubits, with scalability constrained by integer-overflow issues~\cite{bib26}.
Decision-diagram-based property checking has been reported for circuits containing up to 10 qubits~\cite{bib27}, while equivalence-checking approaches based on trusted reference circuits have been evaluated for circuits of up to 128 qubits, relying on Hilbert-space representations whenever binary-space reductions are not applicable~\cite{bib28,bib29}.

Our previous work~\cite{bib34} introduced an abstraction framework for a restricted but practically relevant class of quantum circuits composed solely of Hadamard (H), Pauli-X (X), and C-NOT gates. 
Within this gate set, the superposition status of each qubit and its computational-basis behaviour can be represented independently using two compact bit vectors. 
This representation transforms verification from reasoning over an exponentially large Hilbert space into a polynomial-sized satisfiability problem that can be solved using an off-the-shelf SMT solver. 
The approach was evaluated on Bell-state generation, GHZ-state generation, and the 2-qubit, 3-qubit, and Shor 9-qubit error-detection codes. 
In particular, GHZ-state circuits containing up to 8,192 qubits were verified in less than three minutes.

Although the initial study demonstrated the feasibility of the abstraction, several important questions remained unanswered. 
First, it was unclear whether the abstraction extends beyond the specific code sizes evaluated previously or generalises to arbitrary circuit families.
Second, its applicability to quantum error-correcting codes with more complex parity-check structures, such as CSS codes, had not been investigated.
Third, the correctness lemmas underlying the abstraction were stated without complete proofs, and 
the class of circuits for which the abstractions remain faithful was not characterised. Moreover, the 
effectiveness of the 
properties in detecting implementation faults had not been systematically validated. 
Finally, the extent to which the reported scalability resulted from the abstraction itself, rather than from implementation-specific aspects of the SMT encoding or from the particular circuit structure, required further investigation.

This article addresses these questions and extends our previous work~\cite{bib34}. 
Table~\ref{tab:novelty} gives a detailed comparison between the two studies. 
The main contributions of this article are as follows.

\begin{enumerate}
\item \textbf{Generalisation of correctness properties and circuit families.}
We extend the bit-flip error-detection property from fixed-size repetition codes to arbitrary $n$-qubit repetition codes with $n-1$ syndrome ancillas. 
Similarly, the phase-flip property is extended from fixed instances to arbitrary $k$-block concatenated codes. 
Both circuit families are verified for values of $n$ and $k$ up to 101.\\
\item \textbf{Verification of CSS syndrome extraction.}
We extend the abstraction framework to the $[[7,1,3]]$ Steane code, based on the classical Hamming$(7,4)$ parity-check matrix, and more generally to CSS codes defined by arbitrary parity-check matrices $\mathbf{H}_X$ and $\mathbf{H}_Z$.
Within the restricted $\{H,X,\text{C-NOT}\}$ gate set, we formally verify bit-flip and phase-flip syndrome extraction for the Steane code, the $[[15,7,3]]$ quantum Hamming code, the $[[15,1,3]]$ quantum Reed-Muller code, and rotated surface codes with distances up to 101.\\
\item \textbf{Detection, identification, and recovery.}
Beyond the single-error properties considered in~\cite{bib34}, we verify syndrome extraction for arbitrary error patterns, error detection up to a specified error weight, and recovery using a classical decoder followed by a classically controlled correction.
For each code, we also check the expected boundary cases. 
The positive case is verified, while the negative case produces an undetectable or incorrectly corrected error pattern from the solver.\\
\item \textbf{Formal soundness and complete proofs.}
We define a fragment $\mathcal{F}$ of $\{H,X,\text{C-NOT}\}$ circuits for which the abstraction is faithful to the corresponding Hilbert-space semantics.
We prove soundness on this fragment using an explicit invariant (Theorem~\ref{thm:soundness}) and provide counterexamples showing that each condition defining $\mathcal{F}$ is necessary.
The result is further checked through differential testing against exact state-vector simulation.
We also provide complete proofs of Lemmas~1-6 from~\cite{bib34}, which were previously stated without proof. 
In this article, these lemmas are presented as exact characterisations of the modifications detected by the abstraction.\\
\item \textbf{Empirical validation through fault injection.}
We develop a mutation-testing framework covering 12 circuit families and using 13 mutation operators.
These operators include missing, additional, re-wired, reordered, and control/target-exchanged C-NOT gates, as well as missing and additional Hadamard gates and incorrectly initialised ancillas.
Equivalent mutants are identified automatically using an SMT miter and excluded from the evaluation.
Within the evaluated fault model, all 622 non-equivalent mutants are detected.\\
\item \textbf{Improved implementation and scalability.}
We reimplement the verification framework using the Z3 Python API instead of the SMT-LIB text-generation pipeline used previously.
With this implementation, GHZ-state circuits containing up to 262,144 qubits are verified on a 16~GB laptop for three circuit topologies.
The evaluation includes repeated measurements and records memory usage for each instance.
A controlled ablation using the same hardware indicates that neither the encoding pipeline nor the solver version accounts for the improvement over~\cite{bib34}.
Experiments with random H+C-NOT networks also show that verification becomes increasingly expensive as circuit depth increases.\\
\item \textbf{Scope and limitations of the abstraction.}
We identify the classes of stabilizer codes that can be handled by the proposed abstraction and those that remain outside its current scope.
In particular, we discuss non-CSS stabilizer codes that require genuine $Y$-type stabilizers, including the 5-qubit perfect code, as well as fault-tolerance issues such as hook errors.
We also outline a possible roadmap for extending the framework to these cases.
\end{enumerate}

\begin{table}[!h]
\centering
\caption{Contributions of the preliminary conference paper~\cite{bib34} and of this article.}
\label{tab:novelty}
\footnotesize
\begin{tabularx}{\textwidth}{
>{\raggedright\arraybackslash}p{1.95cm}
>{\raggedright\arraybackslash}X
>{\raggedright\arraybackslash}X
}
\toprule
Aspect & Preliminary work~\cite{bib34} & This article \\
\midrule
Abstraction semantics & Definitions~\ref{def:qubit}-\ref{def:singular} & Unchanged gate semantics. Saturating tag explained, two readings of $q_0$ made explicit, fragment $\mathcal{F}$ defined (Def.~\ref{def:fragment}), classically controlled X for recovery (Def.~\ref{def:ccx}) \\
Correctness properties & Properties~\ref{prop:1}-\ref{prop:3} for fixed sizes & Parametrised Properties~\ref{prop:genrep},~\ref{prop:genshor}; CSS Properties~\ref{prop:steane},~\ref{prop:css}; multi-error, detection and recovery Properties~\ref{prop:anyerr}-\ref{prop:recover} \\
Circuit families & 2-, 3-qubit, Shor 9-qubit codes; Bell; GHZ (star) & $n$-qubit repetition and $k$-block concatenated codes ($n,k\le101$); Steane, $[[15,7,3]]$, $[[15,1,3]]$, surface codes; GHZ with star, chain and tree topologies; random H+C-NOT networks \\
Soundness & Not stated & Theorem~\ref{thm:soundness} on $\mathcal{F}$ with invariant-based proof, necessity counterexamples, differential validation \\
Lemmas 1-6 & Stated without proof & Proved; restated as exact characterisations \\
Mutation testing & None & 13 operators, 12 families, automatic equivalent-mutant identification \\
Implementation & SMT-LIB text generation & Direct Z3 Python API; controlled pipeline and solver-version ablation \\
Evaluation platform & 24-core workstation, 192,GB & Apple M1 laptop, 16,GB \\
Largest GHZ instance & 8,192 qubits & 262,144 qubits (three topologies) \\
\bottomrule
\end{tabularx}
\end{table}

The remainder of the article is organised as follows.
Section~\ref{sec:theory} reviews the background on quantum computing, formal verification, and the classical coding theory needed for CSS codes. 
It also establishes the terminology used for error detection and correction.
Section~\ref{sec:relwrk} discusses related work.
Section~\ref{sec:absfrmwrk} presents the extended abstraction framework, together with the new correctness properties and their proofs.
Section~\ref{sec:methodology} describes the verification methodology and the mutation-testing framework.
Section~\ref{sec:results} presents the experimental results.
Section~\ref{sec:discussion} discusses the applicability, limitations, and threats to validity of the approach.
Finally, Section~\ref{sec8} concludes the article and outlines directions for extending the framework towards the verification of general stabiliser codes.

\section{Theoretical Background}\label{sec:theory}

This section provides the quantum-computing concepts required for the remainder of the article. 
More comprehensive treatments of quantum computation can be found in~\cite{bib13,bib14}.

A quantum bit (qubit) is represented by a unit vector in a two-dimensional complex Hilbert space. Its state can be written as
\[
|\psi\rangle=\alpha|0\rangle+\beta|1\rangle,
\]
where
\[
|0\rangle=\binom{1}{0}, \ \
|1\rangle=\binom{0}{1},
\]
$\alpha,\beta\in\mathbb{C}$, and $|\alpha|^2+|\beta|^2=1$. An $n$-qubit system is described by a state vector in a $2^n$-dimensional Hilbert space,

\begin{equation}
|\Psi\rangle=\sum_{x\in\{0,1\}^n}\alpha_x|x\rangle,
\label{eq:state}
\end{equation}
where the probability of observing basis state $x$ after measurement is given by $\Pr(x)=|\alpha_x|^2$. 
Quantum computation is performed by applying unitary operators, referred to as quantum gates, to the quantum state. 
A quantum circuit is an ordered sequence of such gates that transforms the input state into the desired output state~\cite{bib15}. 
Since the dimension of the state space increases exponentially with the number of qubits, direct analysis of quantum circuits rapidly becomes computationally expensive, making scalable verification particularly challenging.

\subsection{Entanglement: Bell and GHZ states}

Entanglement is one of the distinguishing features of quantum computation. 
A multi-qubit state is said to be entangled if it cannot be expressed as the tensor product of the states of its individual qubits. 
The Bell state,
\begin{equation}
|\Phi^+\rangle=
\frac{1}{\sqrt2}\left(|00\rangle+|11\rangle\right),
\label{eq:bell}
\end{equation}
provides the simplest example of a maximally entangled two-qubit state. 
It is generated by applying a Hadamard gate to one qubit followed by a C-NOT gate controlled by that qubit.

The Greenberger-Horne-Zeilinger (GHZ) state extends this construction to multiple qubits~\cite{bib16},
\begin{equation}
|\mathrm{GHZ}\rangle=
\frac{1}{\sqrt2}
\left(|0\rangle^{\otimes m}+|1\rangle^{\otimes m}\right),
\label{eq:ghz}
\end{equation}
where $m>2$. 
A GHZ state is created by applying a Hadamard gate to a single control qubit followed by a sequence of $m-1$ C-NOT gates that entangle the remaining qubits with the control. 
Because entangled states cannot be decomposed into independent single-qubit states, they present additional challenges for functional verification that are absent in classical circuits~\cite{bib16}.

\subsection{Quantum errors and repetition codes}
\label{subsec:rep_codes}
Current quantum hardware is susceptible to errors originating from imperfect gate operations, decoherence, leakage, and measurement noise~\cite{bib17,bib18}. 
Quantum error correction mitigates these effects by encoding logical information across multiple physical qubits.

The simplest encoding strategy is based on the classical repetition code. 
A logical qubit is redundantly represented by multiple physical qubits, for example,
\[
|0\rangle_L\mapsto|00\rangle,\qquad
|1\rangle_L\mapsto|11\rangle,
\]
for the two-qubit code, and
\[
|0\rangle_L\mapsto|000\rangle,\qquad
|1\rangle_L\mapsto|111\rangle,
\]
for the three-qubit code~\cite{bib19,bib20}. 
Bit-flip errors are detected by computing parity checks between selected pairs of physical qubits using ancilla qubits and C-NOT gates without directly measuring the encoded logical state. 
Table~\ref{tab:syndrome} summarizes the syndrome values for the three-qubit repetition code, where ancilla qubits $q_{1A}$ and $q_{2A}$ compute the parity values $q_0\oplus q_1$ and $q_0\oplus q_2$, respectively.

\begin{table}[h]
\centering
\caption{Syndrome table for the 3-qubit bit-flip code (ancilla value $q_{2A}q_{1A}$).}
\label{tab:syndrome}
\begin{tabular}{ccc}
\toprule
Ancilla value & Qubit state & Inference \\
\midrule
00 & $|000\rangle+|111\rangle$ & No error \\
11 & $|001\rangle+|110\rangle$ & 1st qubit flip \\
01 & $|010\rangle+|101\rangle$ & 2nd qubit flip \\
10 & $|100\rangle+|011\rangle$ & 3rd qubit flip \\
\bottomrule
\end{tabular}
\end{table}

The Hadamard and C-NOT gates used throughout this work are represented by
\begin{equation}
H=
\frac{1}{\sqrt2}
\begin{pmatrix}
1&1\\
1&-1
\end{pmatrix},
\ \
\text{C-NOT}=
\begin{pmatrix}
1&0&0&0\\
0&1&0&0\\
0&0&0&1\\
0&0&1&0
\end{pmatrix}.
\label{eq:gates}
\end{equation}

The Shor nine-qubit code~\cite{bib6} extends the repetition-code principle by combining bit-flip and phase-flip protection. Three repetition-code blocks are first transformed into the Hadamard basis, after which parity checks are performed across the blocks using two additional ancilla qubits. The qubits are then transformed back to the computational basis. A single phase-flip error within any block produces a unique two-bit syndrome analogous to that of the three-qubit repetition code, but at the level of encoded blocks rather than individual qubits.

\subsection{Classical linear codes and the CSS construction}
\label{sec:cssbg}

The repetition codes considered in Section~\ref{subsec:rep_codes} are examples of a broader class of quantum error-correcting codes constructed from classical binary linear codes. 
The Calderbank-Shor-Steane (CSS) construction was introduced by Calderbank and Shor~\cite{bib35} and, independently, by Steane~\cite{bib36}~\cite{bib37}. 
It is the only part of classical coding theory required in this work.

A classical binary linear code $\mathcal{C}\subseteq\{0,1\}^N$ with dimension $K_c$ can be represented by a parity-check matrix $\mathbf{H}\in\{0,1\}^{r\times N}$ such that
$\mathcal{C}=\ker\mathbf{H}$ over GF(2).
If a codeword $x\in\mathcal{C}$ is affected by an error pattern $\mathbf{e}$, the received word is $x\oplus\mathbf{e}$.
Its syndrome is
\[
\mathbf{H}(x\oplus\mathbf{e})=\mathbf{H}\mathbf{e}\pmod 2,
\]
because $\mathbf{H}x=0$.
Thus, the syndrome depends only on the error pattern.
For example, a single error on bit $j$ produces the $j^{th}$ column of $\mathbf{H}$ as its syndrome.
If all columns of $\mathbf{H}$ are distinct and non-zero, as for the Hamming$(7,4)$ code, each single-bit error has a different syndrome.
More generally, a code with minimum distance $d$ detects every non-zero error of weight at most $d-1$ and uniquely corrects every error of weight at most $\lfloor(d-1)/2\rfloor$.

A CSS code on $N$ qubits is described by two parity-check matrices, $\mathbf{H}_X$ and $\mathbf{H}_Z$.
Their rows specify the supports of the $X$-type and $Z$-type stabilizer generators, respectively.
For the stabilizers to commute, the matrices must satisfy
\begin{equation}
\mathbf{H}_X \mathbf{H}_Z^{\mathsf{T}}=0 \pmod 2 .
\label{eq:csscond}
\end{equation}
The number of logical qubits encoded by the code is
\[
K=N-\mathrm{rank}\ \mathbf{H}_X-\mathrm{rank}\ \mathbf{H}_Z.
\]
The logical computational-basis states can be viewed as superpositions of classical words $x\in\ker\mathbf{H}_Z$.
After applying the Hadamard transform, the corresponding description uses words $u\in\ker\mathbf{H}_X$.

This construction gives two classical parity-checking problems within the quantum code.
An $X$-type, or bit-flip, error pattern $\mathbf{e}$ is detected by the $Z$-type checks, which produce the syndrome $\mathbf{H}_Z\mathbf{e}$.
Similarly, a $Z$-type, or phase-flip, error acts as a bit flip in the Hadamard basis and is detected by the $X$-type checks, giving the syndrome $\mathbf{H}_X\mathbf{e}$.
In both cases, C-NOT gates accumulate the parity of the checked data qubits onto syndrome ancillas.
This is the same parity-accumulation mechanism used by the repetition codes and captured by the bit-vector abstraction considered in this article.

For the Steane code, $\mathbf{H}_X=\mathbf{H}_Z=\mathbf{H}$, where $\mathbf{H}$ is the parity-check matrix of the Hamming$(7,4)$ code.
Equation~\eqref{eq:csscond} is satisfied because the Hamming code contains its dual.
The other codes considered in Section~\ref{sec:results} include the $[[15,7,3]]$ quantum Hamming code~\cite{bib42}, the $[[15,1,3]]$ quantum Reed-Muller code~\cite{bib43}, for which $\mathbf{H}_X\neq\mathbf{H}_Z$, and rotated surface codes~\cite{bib44,bib45}.

Not all stabilizer codes have the CSS form.
For example, the five-qubit perfect code~\cite{bib46} contains stabilizer generators with $Y$-type factors and therefore cannot be represented by two separate binary parity-check matrices for the $X$ and $Z$ components.
Such codes are naturally described using additive codes over GF(4), or equivalently as self-orthogonal subspaces of $\mathrm{GF}(2)^{2N}$ under the symplectic inner product~\cite{bib47,bib37}.
Capturing this additional algebraic structure would be necessary to extend the abstraction to general stabilizer codes, as discussed in Section~\ref{sec8}.

\subsection{Error detection, identification, decoding, recovery, and fault tolerance}
\label{sec:detcorr}
The terms used for error handling can have different meanings, so we use the following definitions throughout the article.
\emph{Error detection} means that every error pattern in the specified class produces a non-zero syndrome.
\emph{Error identification} means that the measured syndrome uniquely determines the error pattern, up to multiplication by a stabilizer, within the considered class.
\emph{Decoding} is the classical procedure that maps a measured syndrome to a correction.
\emph{Recovery} is the subsequent application of this correction to the data qubits.
Thus, in this article, \emph{error correction} refers to the complete process of syndrome extraction, decoding, and recovery.

\emph{Fault tolerance} is a stronger requirement.
It requires the syndrome-extraction and recovery circuits themselves to tolerate faults, 
such as errors on ancilla qubits or C-NOT gates, without allowing these faults to spread into uncorrectable errors on the data qubits~\cite{bib6,bib48}.
The properties introduced in~\cite{bib34} establish correct syndrome extraction and identification for single errors.
Section~\ref{sec:multierr} extends these properties to arbitrary error patterns, detection up to a specified error weight, and decoding followed by recovery.
Fault tolerance is not addressed by the present abstraction and remains outside its scope, as discussed in Section~\ref{sec:discussion}.

\subsection{Formal verification}
Formal verification establishes system correctness by proving that a design satisfies its specification over all possible behaviours rather than validating selected execution scenarios. 
It has been widely adopted for hardware and software systems~\cite{bib22,bib23,bib24}. 
Extending formal verification techniques to quantum circuits remains an active area of research because of the exponential growth of the quantum state space and the probabilistic nature of quantum measurement~\cite{bib11,bib25}. 
The work presented in this article builds on the abstraction framework introduced in~\cite{bib34}, which reduces verification of quantum circuits composed of Hadamard, Pauli-X, and C-NOT gates to a bit-vector satisfiability problem that can be solved using an SMT solver~\cite{bib33}. 
The following sections extend this framework to generalized circuit families and CSS quantum error-correcting codes, and evaluate its effectiveness through formal verification and mutation-based validation.

\section{Related Work}\label{sec:relwrk}

Formal verification of quantum circuits has attracted increasing attention as quantum hardware and algorithms continue to advance. 
Existing verification techniques differ in their underlying representations and verification strategies, resulting in different trade-offs between expressiveness and scalability. 
The approaches most closely related to this work can be grouped into rewrite-rule-based verification, decision-diagram methods, equivalence checking, symbolic abstractions, and quantum program logics.

Rewrite-rule-based verification has been investigated by Amy~\cite{bib26}, who proposed a framework based on dyadic arithmetic and rewrite rules implemented within a Haskell theorem prover. 
The method supports verification of universal quantum circuits but relies on dyadic representations whose arithmetic becomes increasingly difficult to manage as circuit size grows. 
Reported experiments demonstrate verification for circuits containing up to 96 qubits, with larger instances limited by integer-overflow issues.

Decision-diagram representations have also been explored for quantum circuit verification. 
Seiter \emph{et al.}~\cite{bib27} employed multiple-valued decision diagrams to verify properties of quantum circuits, including implementations of Grover's algorithm. 
Although decision diagrams provide compact representations for many classes of circuits, their effectiveness decreases as the amount of entanglement increases, and the reported evaluation is limited to circuits containing up to 10 qubits.

A different line of research focuses on equivalence checking, in which the implementation under verification is compared with a trusted reference circuit.
Burgholzer and Wille~\cite{bib28} and Yamashita and Markov~\cite{bib29} developed equivalence-checking techniques for quantum circuits using this strategy. 
These approaches assume the availability of a correct reference implementation and rely on Hilbert-space representations when binary-space reductions cannot be applied. 
Experimental evaluations reported in these studies include circuits containing up to 128 qubits.

Symbolic abstractions provide an alternative by reducing quantum-circuit verification to reasoning over compact symbolic representations. 
Our earlier work introduced symbolic, rotational, and superposition-based abstractions for quantum circuits represented as bit vectors~\cite{bib30,bib31,bib32,bib39}. 
These abstractions were applied to the verification of the Quantum Fourier Transform, its inverse, quantum data-encoding circuits, and quantum phase-estimation circuits. 
The abstraction framework presented in~\cite{bib34} extended this approach to entanglement-generation circuits and quantum error-detection circuits composed of Hadamard, Pauli-X, and C-NOT gates. 
That study demonstrated the feasibility of the abstraction for Bell-state generation, GHZ-state generation, and the 2-qubit, 3-qubit, and Shor 9-qubit error-detection codes, but it considered only fixed code sizes and presented the underlying correctness lemmas without formal proofs.

A closely related, but different, line of work concerns the efficient \emph{simulation} of stabilizer circuits. 
The Gottesman-Knill theorem~\cite{bib50} shows that circuits composed of Clifford gates, including $H$, $X$, and C-NOT, together with computational-basis measurements, 
can be simulated classically in polynomial time. 
The tableau method of Aaronson and Gottesman~\cite{bib40} requires $O(n^2)$ bits for $n$ qubits, with $O(n)$ time per gate and $O(n^2)$ time per measurement. 
Stim~\cite{bib41} provides a highly optimised implementation that is widely used in quantum error-correction studies.
Simulation and the SMT-based verification considered in this article serve different purposes.
A simulator evaluates the behaviour of a circuit for a particular input, or samples its behaviour under sampled noise.
In contrast, an SMT query checks whether a property holds for all admissible inputs and error patterns, or produces a counterexample when it does not.
Error patterns are represented symbolically, so a single query can cover, for example, all $2^N$ error patterns considered by Property~\ref{prop:anyerr}.
The scalability results in Section~\ref{sec:results} should therefore not be interpreted as showing that the proposed abstraction is more efficient than stabilizer simulation.
Stabilizer simulation applies in polynomial time to the full Clifford group, whereas the abstraction considered here is limited to the fragment $\mathcal{F}$ defined in Definition~\ref{def:fragment}.
Instead, the results show that the universally quantified verification queries generated by the abstraction remain tractable for the circuit families evaluated in this study.

Formal reasoning about quantum programs has also been studied from the perspective of program logics. 
Liu \emph{et al.}~\cite{bib10} proposed a quantum Hoare logic for verifying the correctness of quantum algorithms, including amplitude-amplification procedures. 
This line of work is complementary to circuit-level functional verification, as it focuses on proving algorithmic correctness rather than verifying the structural correctness of specific quantum circuit implementations. 
Broader surveys of formal verification techniques for quantum programs and circuits have been presented by Chareton \emph{et al.}~\cite{bib11}, Quist \emph{et al.}~\cite{bib25}, and Govindankutty~\cite{bib12}, highlighting the diversity of existing approaches and the continuing challenges associated with scalable verification.

The work in this article directly extends the abstraction framework introduced in~\cite{bib34}.
Table~\ref{tab:novelty} in Section~\ref{sec:intro} gives a detailed comparison of the material carried over from~\cite{bib34} and the new contributions.
The present work extends the verification methodology to arbitrary repetition-code families and to CSS codes defined by arbitrary parity-check matrices.
It also defines the circuit fragment for which the abstraction is sound and provides a proof of soundness, together with a systematic evaluation based on mutation-based fault injection.
Mutation testing has been used for many years to evaluate verification and testing techniques in software and hardware systems, following the original work of DeMillo, Lipton, and Sayward~\cite{bib38}; a survey is given in~\cite{bib52}.
One challenge in mutation testing is the equivalent-mutant problem: 
a mutant may be functionally identical to the original system and therefore cannot be distinguished from it by any test.
In general, determining equivalence is undecidable~\cite{bib52,bib53}.
For the finite-state abstract circuits considered in this article, however, equivalence is decidable.
We therefore identify equivalent mutants automatically using an SMT miter, as described in Section~\ref{sec:mutationmethod}.

\section{Verification Framework}
\label{sec:absfrmwrk}

The verification framework adopted in this work builds upon the abstraction introduced in our earlier study~\cite{bib34}. 
Sections~\ref{sec:baseabs}, and~\ref{sec:baseprop} briefly summarise the underlying abstraction and its three original correctness properties. 
The remainder of this section presents the extensions introduced in this work: 
 generalised correctness properties for parametrised code families (Section~\ref{sec:genprop}), 
the formalisation of the Steane code, and of general CSS syndrome extraction (Section~\ref{sec:steane}), 
properties for multiple simultaneous errors, detection, and recovery (Section~\ref{sec:multierr}),
and the soundness theorem together with complete proofs of the previously stated lemmas (Sections~\ref{sec:soundness}--\ref{sec:proofs}).

\subsection{Base abstraction}\label{sec:baseabs}
The proposed verification framework represents quantum circuits using a compact bit-vector abstraction. 
Each qubit is modelled by two bit vectors that capture the information required for reasoning about superposition and computational-basis behaviour under the restricted gate set considered in this work.

\begin{definition}[Abstract Qubit]\label{def:qubit}
An abstract qubit $Q$ is represented by the tuple $(s,q)$, where $s$ is a 3-bit vector encoding the qubit's superposition status and $q$ is a 2-bit vector encoding its measurement status and computational-basis value. 
The least significant bit of $q$ stores the basis value, whereas the most significant bit records the measurement status ($0$ for unmeasured and $1$ for measured).
\end{definition}
The abstraction models the action of the Hadamard, Pauli-X, and C-NOT gates directly on these bit-vector representations.

\begin{definition}[Abstract H Gate]\label{def:h}
The abstract Hadamard gate is defined as
\[
H(Q(s,q))=
\begin{cases}
\big((s+b'001)\bmod 2^{3},\ q\big), & \text{if } s\in\{b'000,\dots,b'110\},\\
(b'110,\ q), & \text{if } s=b'111.
\end{cases}
\]
Here, $+$ denotes unsigned integer addition of 3-bit vectors.
Since the first case is used only when $s\le b'110$, the modulo operation does not cause a wrap-around.
The tag therefore records up to seven applications of the Hadamard gate.
The gate changes only the superposition tag $s$ and leaves the basis-value bit $q$ unchanged.
An even value of $s$ represents a qubit that is not in superposition, while an odd value represents a qubit that has undergone an odd number of Hadamard operations and is therefore in superposition.

When $s=b'111$, the tag is saturated rather than allowed to wrap around.
Further Hadamard gates make the tag alternate between $b'110$ and $b'111$.
This preserves the parity of $s$, which is the information used by the semantics in Theorem~\ref{thm:soundness}, regardless of the number of Hadamard gates applied.
It also keeps $s=b'000$ as the unique value representing a qubit to which no Hadamard gate has been applied.
This is important because the correctness properties use $s=b'000$ to express the absence of Hadamard gates.
Without saturation, eight Hadamard applications would produce the same tag as zero applications.
In the benchmark circuits, $s$ never exceeds $b'010$; the value $s=b'011$ occurs only in mutants that contain an additional Hadamard gate.
\end{definition}

\begin{definition}[Abstract X Gate]\label{def:x}
The abstract Pauli-X gate is defined as
\[
X(Q(s,q))=Q(s,q_1\lnot q_0),
\]
where the basis-value bit $q_0$ is complemented, while the superposition vector $s$ and the measurement-status bit $q_1$ remain unchanged.
\end{definition}

\begin{definition}[Abstract C-NOT Gate]\label{def:cnot}
For a control-target pair $(Q_c,Q_t)$,
\[
C(Q_c,Q_t)=
\begin{cases}
(Q_c,Q_t), & \text{if } q_{c0}=0,\\
(Q_c,Q(s_t,q_{t1}\lnot q_{t0})), & \text{if } q_{c0}=1.
\end{cases}
\]
The control qubit is unaffected by the operation. 
The basis-value bit of the target qubit is complemented if and only if the control qubit has basis value 1.
\end{definition}

\paragraph{Two readings of the basis-value bit.}
For a qubit with an even tag, $q_0$ represents its computational-basis value.
For a qubit with an odd tag, the meaning of $q_0$ depends on how the abstraction is being used.
The soundness result in Theorem~\ref{thm:soundness} specifies the conditions under which each interpretation is valid.

In the \emph{branch reading}, a qubit enters superposition when a Hadamard gate is applied to a free input variable and can then participate in C-NOT or $X$ gates.
In this case, $q_0$ represents the basis value of the qubit in each branch of the superposition.
Universal quantification over the free input variable therefore represents both branches.
This is the interpretation used for the Bell and GHZ circuits.

In the \emph{frame reading}, a qubit is placed between two Hadamard gates and does not participate in any other gate.
Here, $q_0$ represents the label of the state in the Hadamard basis $\{|+\rangle,|-\rangle\}$.
Complementing $q_0$ therefore represents a phase-flip ($Z$) error, following the interpretation used in~\cite{bib34}.
This is the interpretation used for phase-flip detection.

The two readings cannot be used interchangeably.
For example, under the frame reading, a physical bit-flip ($X$) leaves the states $|+\rangle$ and $|-\rangle$ unchanged up to a global phase.
However, Definition~\ref{def:x} would complement $q_0$ and therefore represent the operation as a phase flip.
Such a gate is consequently excluded from the circuit fragment for which Theorem~\ref{thm:soundness} is established (Definition~\ref{def:fragment}).
The corresponding counterexample is given in Remark~\ref{rem:cex}.

In several correctness properties presented later in this section, it is necessary to determine whether a bit vector contains exactly one asserted bit. This is captured by the following auxiliary function.

\begin{definition}[Singular Function]\label{def:singular}
For a bit vector $m=(m_0,\dots,m_{n-1})$, the singular function is defined as
\[
S(m)=1
\]
if and only if exactly one element $m_i$ is equal to 1. Otherwise,
\[
S(m)=0.
\]
\end{definition}

\subsection{Base correctness properties}\label{sec:baseprop}

The abstraction described in the Section~\ref{sec:baseabs} is used to express the expected functional behaviour of the quantum circuits  considered in this work. 
Following~\cite{bib34}, three correctness properties are defined to capture the behaviour of bit-flip error-detection circuits, phase-flip error-detection circuits, and entanglement-generation circuits, respectively.

\begin{property}[Bit-flip Correctness]\label{prop:1}
Assume that all ancilla and logical qubits are initially in the computational basis, i.e.,
$s^{in}_{Ai}=s^{in}_{Li}=b'000$. The circuit satisfies the bit-flip correctness property if the following conditions hold:
\begin{enumerate}
\item The superposition status of both ancilla and logical qubits is preserved throughout the circuit, such that
$s^{out}_{Ai}=s^{out}_{Li}=b'000$.
\item Every ancilla qubit is initialized to
$q^{in}_{Ai}=b'00$ and remains unmeasured until the final measurement stage.
\item The measured ancilla values encode the parity relationships
\[
q^{out}_{A0}=b'11 \Leftrightarrow q_{L0_0}\neq q_{L1_0},
\]
and
\[
q^{out}_{A1}=b'11 \Leftrightarrow q_{L0_0}\neq q_{L2_0}.
\]
\end{enumerate}
\end{property}

\begin{property}[Phase-flip Correctness]\label{prop:2}
Assume that all ancilla and logical qubits initially satisfy
$s^{in}_{Ai}=s^{in}_{Li}=b'000$. The phase-flip detection circuit is correct if the following conditions are satisfied:

\begin{enumerate}
\item Ancilla qubits remain outside superposition throughout the computation, yielding
$s^{out}_{Ai}=b'000$, while each logical qubit undergoes exactly one Hadamard transformation at the beginning and one at the end of its computational path, resulting in
$s^{1}_{Li}=b'001$ and
$s^{out}_{Li}=b'010$.
\item Every ancilla qubit is initialized to
$b'00$ and remains unmeasured until the final measurement stage.
\item If exactly one of the logical qubits
$q_{L0},\dots,q_{L8}$ is affected by a flip, the resulting two-bit ancilla syndrome uniquely identifies the logical block containing the error, using the same syndrome mapping as Property~\ref{prop:1} applied at the block level.
\end{enumerate}
\end{property}

\begin{property}[Entanglement Correctness]\label{prop:3}
Consider a quantum system
$Q=\{Q_1,\dots,Q_m\}$. The circuit correctly generates entanglement if the following conditions hold:
\begin{enumerate}
\item The first qubit ($Q_1$) is initially outside superposition, with
$s^{in}_1=b'000$, and enters superposition after the application of a single Hadamard gate, resulting in
$s^1_1=b'001$.
\item No Hadamard gate is applied to any remaining qubit, and therefore the superposition tag of each qubit
$Q_i$, $i>1$, remains
$b'000$ throughout the circuit.
\item For
$q^{in}_1\in{b'00,b'01}$ and
$q^{in}_i=b'00$ for all
$i\neq1$, the circuit produces identical output basis values for all qubits, namely
\[
q^{out}_1=q^{out}_2=\cdots=q^{out}_m.
\]
\end{enumerate}
\end{property}

\subsection{Generalised properties for parametrised code families}\label{sec:genprop}

The correctness properties introduced above describe fixed circuit instances corresponding to the 2-qubit, 3-qubit, and 9-qubit codes studied in~\cite{bib34}. 
To support verification of scalable circuit families, these properties are extended here to parametrised repetition and concatenated-code constructions. 
Their verification results are presented in Section~\ref{sec:results}.

\begin{definition}[Generalised $n$-qubit repetition code]\label{def:genrep}
Consider an $n$-qubit repetition code, where $n\geq2$, with data qubits
$q_0,\dots,q_{n-1}$ sharing a common pre-error basis value $c$ and containing at most one bit-flip error $e_k$ on qubit $k$ ($0\leq k\leq n-1$). 
Define $n-1$ ancilla qubits
$a_i$ ($0\leq i\leq n-2$), where each ancilla computes the parity
\[
a_i\leftarrow q_i\oplus q_{i+1},
\]
using two C-NOT operations as specified in Definition~\ref{def:cnot}.
For $n=3$ this uses adjacent pairs, $(q_0\oplus q_1,\,q_1\oplus q_2)$, rather than the pairs $(q_0\oplus q_1,\,q_0\oplus q_2)$ of Table~\ref{tab:syndrome}. Both choices identify every single error uniquely.
\end{definition}

\begin{property}[Generalised Bit-flip Correctness]\label{prop:genrep}
For the generalised repetition code defined in Definition~\ref{def:genrep}, assuming at most one bit-flip error, the following conditions must hold:

\begin{enumerate}
\item No Hadamard gate is applied within the detection circuit, and therefore every superposition tag remains
$b'000$ throughout the computation.
\item All ancilla qubits are initialised to
$b'00$ and are measured only after syndrome computation has completed.
\item An ancilla value satisfies
$a_i=1$ if and only if the error occurs at position
$k\in\{i,i+1\}$; otherwise,
$a_i=0$. 
Consequently, the syndrome forms a single domain wall centred on the error location, 
extending the syndrome behaviour shown in Table~\ref{tab:syndrome} for the case
$n=3$ to arbitrary values of $n$.
\end{enumerate}
\end{property}

\begin{definition}[Generalised $k$-block concatenated code]\label{def:genshor}
Consider a concatenated code with $k\geq3$ logical blocks.
Each block is represented by one abstract block-qubit carrying the common Hadamard-basis label $c$ of the code word.
Each block-qubit is first acted on by a Hadamard gate, may then undergo a phase-flip error, represented as an abstract bit flip when its superposition tag is $b'001$, following~\cite{bib34}, and is finally returned to the computational basis by a second Hadamard gate.
The circuit then uses $k-1$ ancilla qubits to compute the parity between adjacent blocks:
\[
a_i\leftarrow b_i\oplus b_{i+1},\qquad 0\leq i\leq k-2.
\]
Each parity is computed using two C-NOT gates (Definition~\ref{def:cnot}).
Thus, the resulting syndrome is the repetition-code syndrome from Definition~\ref{def:genrep}, applied at the block level.
\end{definition}

\begin{property}[Generalised Phase-flip Correctness]\label{prop:genshor}
For the generalised concatenated code in Definition~\ref{def:genshor}, assume that at most one logical block contains an error.
The following conditions then hold:\\
\begin{enumerate}
\item Each block-qubit undergoes exactly two Hadamard operations and therefore has superposition tag $b'010$ at the end of the circuit.
The ancilla qubits retain the tag $s=b'000$.\\
\item All ancilla qubits are initialised to $b'00$ and are measured after the parity checks have been computed.\\
\item The syndrome bit satisfies
\[
a_i=1
\quad\Longleftrightarrow\quad
\text{block $i$ or block $i+1$ contains the error}.
\]
Thus, the syndrome is all zero when there is no error and uniquely identifies the erroneous block when exactly one block is affected.
\end{enumerate}
\end{property}

The phase-flip property for the 9-qubit Shor code in Property~\ref{prop:2} is a special case of Property~\ref{prop:genshor} with $k=3$.
In this case, the circuit has two ancilla qubits, giving the same syndrome representation as in~\cite{bib34}.

\subsection{The Steane $[[7,1,3]]$ CSS code} and general CSS syndrome extraction
\label{sec:steane}

The generalised repetition and concatenated-code properties introduced in Section~\ref{sec:genprop} extend the verification framework to 
parametrised circuit families while retaining the underlying repetition-code structure. 
This subsection further extends the framework by formalising the $[[7,1,3]]$ Steane code, 
a Calderbank-Shor-Steane (CSS) quantum error-correcting code, within the existing $\{H,X,\text{C-NOT}\}$ abstraction, and then the construction of arbitrary CSS parity-check matrices.

Let $\mathbf{H}\in\{0,1\}^{3\times7}$ denote the parity-check matrix of the classical Hamming$(7,4)$ code. 
Column $i+1$, corresponding to qubit $i$ ($i=0,\ldots,6$), is given by the 3-bit binary representation of $i+1$:

\begin{equation}
\mathbf{H}=
\begin{pmatrix}
0&0&0&1&1&1&1\\
0&1&1&0&0&1&1\\
1&0&1&0&1&0&1
\end{pmatrix},
\label{eq:hamming}
\end{equation}
where the rows represent the three parity checks, indexed by
$j\in\{0,1,2\}$ from the bottom row (least significant bit of the column oindex) to the top. 
For each parity check, the participating qubits are defined by
$\mathrm{row}_j=\{i:\mathbf{H}_{j,i}=1\}$, yielding
$\mathrm{row}_0=\{0,2,4,6\}$,
$\mathrm{row}_1=\{1,2,5,6\}$, and
$\mathrm{row}_2=\{3,4,5,6\}$.

\begin{definition}[Steane bit-flip ancilla]\label{def:steanebf}
Consider data qubits $q_0,\ldots,q_6$ satisfying the valid-codeword condition
\[
\bigoplus_{i\ \in \ \mathrm{row}_j}q_{i0}=0,
\]
for every parity check $j$, and assume that at most one bit-flip error $e_k$ is present. 
Three ancilla qubits, $b_0$, $b_1$, and $b_2$, are computed according to
\[
b_j\leftarrow\bigoplus_{i \ \in \ \mathrm{row}_j}q_i,
\]
where each parity computation is implemented by a sequence of
$|\mathrm{row}_j|=4$ C-NOT gates (Definition~\ref{def:cnot}). 
This construction follows the same abstraction used in Definition~\ref{def:genrep}, with the parity-check rows of $\mathbf{H}$ replacing the adjacent-qubit parity relations.
\end{definition}

\begin{definition}[Steane phase-flip ancilla]\label{def:steanepf}
The phase-flip detection circuit uses the same construction as Definition~\ref{def:steanebf}.
Each data qubit is first acted on by a Hadamard gate, giving it the superposition tag $s=b'001$.
A phase-flip error is then represented by an abstract bit flip in this representation, following Property~\ref{prop:2} and the frame reading described in Section~\ref{sec:baseabs}.
A second Hadamard gate returns each data qubit to the computational basis, giving $s=b'010$.
The parity checks specified by the rows of $\mathbf{H}$ are then computed using three ancilla qubits.

The relation $HZH=X$ means that a phase flip occurring between the two Hadamard gates is represented as a bit flip after the second Hadamard gate.
The ancilla parities therefore measure the corresponding $X$-type stabilizers of the Steane code.
\end{definition}

\begin{property}[Steane CSS Correctness]\label{prop:steane}
Under Definitions~\ref{def:steanebf} and~\ref{def:steanepf}, 
assume that at most one bit-flip error occurs at qubit position $k$, 
or equivalently that at most one phase-flip error occurs at position $k$.
The resulting three-bit syndrome
$(b_2,b_1,b_0)$ is equal to $0$ in the absence of an error and is equal to the binary representation of $k+1$, 
corresponding to column $k$ of the parity-check matrix $\mathbf{H}$, when qubit $k$ is in error. 
Additionally, the data tags are $b'000$ (bit-flip circuit) or $b'010$ (phase-flip circuit), and the ancillas 
 satisfy clauses~1 and~2 of Property~\ref{prop:genrep}.
\end{property}

Property~\ref{prop:steane} directly reflects the single-error-correction capability of the classical Hamming$(7,4)$ code within the proposed bit-vector abstraction. 
Each ancilla computes the XOR parity of the basis values of the participating data qubits, as defined by Definition~\ref{def:cnot}. 
Since XOR is a linear operation, the syndrome associated with a corrupted codeword can be expressed as the XOR of the syndrome of the original codeword and the syndrome of the error. 
Because a valid codeword has zero syndrome, the resulting syndrome is equal to the corresponding column of $\mathbf{H}$ associated with the erroneous qubit, namely
\[
(\text{syndrome of codeword})\oplus(\text{syndrome of error})
=
0\oplus(\text{column }k\text{ of }\mathbf{H}).
\]

Nothing in Definitions~\ref{def:steanebf} and~\ref{def:steanepf} depends on the particular choice of the Hamming matrix.
For an arbitrary CSS code (Section~\ref{sec:cssbg}), the same construction can be used with $\mathbf{H}_Z$ for bit-flip errors and $\mathbf{H}_X$ for phase-flip errors.

\begin{definition}[CSS syndrome extraction]\label{def:css}
Let $\mathbf{H}_X$ and $\mathbf{H}_Z$ satisfy Eq.~\eqref{eq:csscond}.
For the \emph{bit-flip circuit}, the data labels $v\in\{0,1\}^N$ satisfy $\mathbf{H}_Z v=0$.
A bit-flip error pattern $\mathbf{e}$ is then introduced using the abstract $X$ gate.
For each row of $\mathbf{H}_Z$, one ancilla is used to accumulate the corresponding parity through $|\mathrm{row}_j|$ C-NOT gates.

For the \emph{phase-flip circuit}, the Hadamard-basis labels satisfy $\mathbf{H}_Xv=0$.
Each data qubit is surrounded by two Hadamard gates, with a possible phase-flip error between them, as in Definition~\ref{def:steanepf}.
One ancilla is used for each row of $\mathbf{H}_X$ to accumulate the corresponding parity.
\end{definition}

\begin{property}[CSS Syndrome Correctness]\label{prop:css}
Consider the circuits in Definition~\ref{def:css}, with $\mathbf{H}$ equal to either $\mathbf{H}_Z$ or $\mathbf{H}_X$, and assume that at most one error occurs.
The measured ancilla vector is then $\mathbf{H}\mathbf{e}$.
The data tags are $b'000$ for the bit-flip circuit and $b'010$ for the phase-flip circuit.
The ancillas also satisfy clauses~1 and~2 of Property~\ref{prop:genrep}.
For the Steane code, this property reduces to Property~\ref{prop:steane}.
\end{property}

\subsection{Multiple errors, detection, and recovery}\label{sec:multierr}
The properties above consider at most one error and show that the resulting syndrome identifies that error.
We also consider multiple errors in order to distinguish the different notions introduced in Section~\ref{sec:detcorr}.
Only the first property below requires no additional gate semantics.

\begin{property}[Syndrome Correctness under Arbitrary Errors]\label{prop:anyerr}
For the repetition, concatenated, and CSS circuits, consider every error pattern $\mathbf{e}\in\{0,1\}^N$, without restricting its weight.
The measured ancilla vector must equal $\mathbf{H}\mathbf{e}$ modulo 2.
The corresponding single-error property must also satisfy clauses~1 and~2, and the data basis values after syndrome extraction must be the same as after error injection.
Thus, the syndrome-extraction circuit must not introduce additional changes to the data.

For the repetition code, the syndrome is
\[
\mathbf{H}\mathbf{e}=(e_i\oplus e_{i+1})_i.
\]
The single-error properties constrain the ancillas and the superposition tags, whereas this property also requires preservation of the data values.
\end{property}

\begin{property}[Detection up to Weight $w$]\label{prop:detect}
Every error pattern with
\[
1\leq|\mathbf{e}|\leq w
\]
must produce a non-zero syndrome.
\end{property}

Recovery requires one additional operation that depends only on the measured syndrome.

\begin{definition}[Classically-controlled X]\label{def:ccx}
Let $\kappa$ be a Boolean condition on the measured ancilla values.
The classically controlled operation $X^{\kappa}$ is defined by
\[
X^{\kappa}(Q(s,q))=
\begin{cases}
X(Q(s,q)), & \text{if }\kappa\text{ holds},\\
Q(s,q), & \text{otherwise}.
\end{cases}
\]
The operation is applied only to qubits with an even tag.
\end{definition}

\begin{property}[Recovery up to Weight $w$]\label{prop:recover}
After syndrome extraction, a classical decoder uses the measured syndrome to determine a correction.
The correction is then applied using the operation defined in Definition~\ref{def:ccx}.
For every error pattern with $|\mathbf{e}|\leq w$, the corrected data labels must match the error-free labels: $c$ for the repetition code and $v$ for the CSS codes.

For the repetition code, we use a minimum-weight decoder.
It reconstructs the error pattern from the domain walls and complements the reconstructed pattern when its weight is greater than $n/2$.
For CSS codes whose check matrix has distinct non-zero columns, we use a look-up decoder that corrects qubit $j$ when the measured syndrome is equal to column $j$ of the check matrix.
For the phase-flip circuit, the correction is applied after the second Hadamard gate.
At this point, an $X$ correction represents a $Z$ correction on the encoded data.
\end{property}

For a code with distance $d$, Property~\ref{prop:detect} is expected to hold for $w=d-1$ and to fail for $w=d$.
Similarly, Property~\ref{prop:recover} is expected to hold for
$w=\lfloor(d-1)/2\rfloor$ and to fail for
$w=\lfloor(d-1)/2\rfloor+1$.
Section~\ref{sec:resmulti} verifies both the positive and negative cases.
For the negative cases, the solver produces an error pattern that is either undetectable or incorrectly corrected.

\subsection{Soundness of the abstraction}
\label{sec:soundness}

The correctness properties introduced above are defined entirely in terms of the proposed bit-vector abstraction.
Before using these properties for verification, we need to establish when the abstraction correctly represents the underlying quantum circuit.
The abstraction is not faithful for every circuit over $\{H,X,\text{C-NOT}\}$ (Remark~\ref{rem:cex}).
We therefore first define the circuit fragment for which the soundness result holds.

\begin{definition}[Fragment $\mathcal{F}$]\label{def:fragment}
Consider a circuit over $\{H,X,\text{C-NOT}\}$ with computational-basis inputs, where phase-flip errors may also be injected.
The circuit belongs to $\mathcal{F}$ if every qubit has one of the following three roles.

\begin{enumerate}
\item \emph{Computational qubits} receive no Hadamard gate.\\
\item \emph{Branch qubits} receive exactly one Hadamard gate, and this gate is applied before the qubit participates in any C-NOT gate.
The input basis value of the qubit is a free variable in the verification query.
After the Hadamard gate, the qubit may participate in arbitrary $X$ and C-NOT gates.\\
\item \emph{Frame qubits} receive an even number of Hadamard gates arranged in consecutive pairs.
Between the two Hadamard gates in each pair, called a \emph{frame segment}, 
the qubit does not participate in any gate other than an injected phase-flip error, represented by complementing $q_0$.
Outside these frame segments, the qubit may participate in arbitrary $X$ and C-NOT gates.
\end{enumerate}

Measurements (Definition~\ref{def:qubit}) and classically controlled $X$ gates (Definition~\ref{def:ccx}) are applied only to qubits with an even tag, and only after all gates acting on those measured qubits have been completed.
\end{definition}

Membership in $\mathcal{F}$ is a syntactic condition and can be checked by scanning the gate list once.
All circuits evaluated in this article belong to $\mathcal{F}$.
The bit-flip repetition and CSS circuits use only computational qubits.
The concatenated and CSS phase-flip circuits use frame qubits with one frame segment.
The Bell and GHZ circuits use one branch qubit.
The random networks in Section~\ref{sec:resrandom} use both branch and computational qubits.

\begin{theorem}[Soundness on $\mathcal{F}$]\label{thm:soundness}
Let $\mathcal{C}\in\mathcal{F}$ be a circuit on $n$ qubits, let $B$ be its set of branch qubits, and let $\widehat{\mathcal{C}}$ be its abstraction according to Definitions~\ref{def:qubit}--\ref{def:cnot}.
For input basis values $x\in\{0,1\}^n$ and an assignment $\mathbf{e}$ of the error indicators, let $y(x,\mathbf{e})\in\{0,1\}^n$ and $s(x,\mathbf{e})$ denote the output basis bits $q_0$ and tags produced by $\widehat{\mathcal{C}}$.
Let $\mathcal{C}_{\mathbf{e}}$ denote the corresponding concrete circuit in which each injected phase-flip error is represented by a $Z$ gate.
Then, before measurement,
\begin{equation}
\mathcal{C}_{\mathbf{e}}|x\rangle
=
2^{-|B|/2}
\!\!\sum_{x'\,:\,x'_i=x_i\ \forall i\notin B}
\!\!
\omega(x')\,\bigl|y(x',\mathbf{e})\bigr\rangle,
\qquad |\omega(x')|=1,
\label{eq:soundness}
\end{equation}
where the $2^{|B|}$ basis states $\bigl|y(x',\mathbf{e})\bigr\rangle$ are pairwise distinct.
Moreover, the tag of qubit $i$ is odd if and only if $i\in B$, and its value equals the number of Hadamard gates applied to that qubit whenever this number is at most seven.
\end{theorem}

\begin{proof}
Consider a prefix containing the first $t$ gates of the circuit.
Let $B_t\subseteq B$ be the branch qubits whose Hadamard gate has already occurred, and let $O_t$ be the frame qubits that are currently inside an open frame segment.
For $z\in\{0,1\}^n$, let $|z\rangle_{O_t}$ denote the product state in which each qubit $i\in O_t$ is in the Hadamard-basis state $H|z_i\rangle$, while all other qubits are in the computational-basis state $|z_i\rangle$.

We prove by induction on $t$ the invariant
\[
\mathcal{I}(t):
\quad
\mathcal{C}^{(t)}_{\mathbf{e}}|x\rangle
=
2^{-|B_t|/2}
\!\!\sum_{x'\,:\,x'_i=x_i\ \forall i\notin B_t}
\!\!
\omega_t(x')\,|y^{(t)}(x',\mathbf{e})\rangle_{O_t},
\]
where $y^{(t)}$ denotes the abstract basis bits after the prefix.
The terms in the sum are pairwise distinct, and the tag of qubit $i$ is odd exactly when $i\in B_t\cup O_t$.

\emph{Base case.}
For $t=0$, no gates have been applied, so $B_0=O_0=\emptyset$.
All tags are $b'000$, and both sides of the invariant are equal to $|x\rangle$.

\emph{Inductive step.}
Assume that $\mathcal{I}(t)$ holds and consider the next gate.
By linearity, it is enough to consider one term $|z\rangle_{O_t}$, where
$z=y^{(t)}(x',\mathbf{e})$.

\begin{enumerate}
\item \emph{$X$ gate on qubit $i\notin O_t$.}
The gate maps
\[
X|z_i\rangle=|z_i\oplus1\rangle,
\]
which is exactly the update specified by Definition~\ref{def:x}.
The tag is unchanged.

\item \emph{C-NOT with control $c$ and target $u$, both outside $O_t$.}
By Definition~\ref{def:fragment}, a C-NOT cannot act on a qubit inside an open frame segment.
For the remaining qubits,
\[
\mathrm{C\mbox{-}NOT}|z_c,z_u\rangle
=
|z_c,z_u\oplus z_c\rangle,
\]
which agrees with Definition~\ref{def:cnot}.
The $X$ and C-NOT gates only permute basis states, so distinct terms remain distinct.

\item \emph{Hadamard gate on a branch qubit $i$.}
The qubit has not previously participated in a C-NOT gate.
Therefore its value can be written as $z_i=x_i\oplus\beta$, where $\beta$ is the parity of the $X$ gates applied to it.
The variable $x_i$ does not occur in any other component of $z$.
Hence
\[
H|z_i\rangle
=
2^{-1/2}
\sum_{r\in\{0,1\}}
(-1)^{z_ir}|r\rangle.
\]
The abstraction retains $q_0=z_i$, while $x'_i$ ranges over both values.
The resulting sum is therefore represented by the new free variable $x'_i$.
The phase factor $(-1)^{z_ir}$ is absorbed into $\omega_{t+1}$, and
$B_{t+1}=B_t\cup\{i\}$.
The tag becomes odd.

\item \emph{Hadamard gate opening a frame segment on qubit $i$.}
The state changes from
\[
|z_i\rangle\longmapsto H|z_i\rangle.
\]
Thus, qubit $i$ enters $O_{t+1}$ with Hadamard-basis label $z_i$.
The abstraction retains $q_0=z_i$, and the tag becomes odd.

\item \emph{Phase-flip error on $i\in O_t$.}
Using $ZH=HX$, we obtain
\[
ZH|z_i\rangle
=
HX|z_i\rangle
=
H|z_i\oplus1\rangle.
\]
Thus, the phase-flip complements the Hadamard-basis label, which is exactly how the error is represented in the abstraction.

\item \emph{Hadamard gate closing a frame segment on qubit $i$.}
The two Hadamard gates cancel:
\[
HH|z_i\rangle=|z_i\rangle.
\]
The qubit therefore leaves $O_{t+1}$ with computational-basis value $z_i=q_0$, and its tag becomes even.
\end{enumerate}

In every case, the invariant holds for $t+1$.
The statement about the tags follows from Definition~\ref{def:h}: each Hadamard gate increments the tag, while saturation preserves its parity.
At the end of the circuit, all frame segments have been closed, so $O_t=\emptyset$ and $B_t=B$.
This gives Eq.~\eqref{eq:soundness}.
\end{proof}

\begin{corollary}\label{cor:props}
Let $P$ be a property of the abstract components $(s,q)$ of a circuit $\widehat{\mathcal{C}}$, where $\mathcal{C}\in\mathcal{F}$.
Suppose that the SMT query for $\neg P$ is unsatisfiable, so that $P$ holds for every assignment of the free input and error variables allowed by the query preconditions.
Then, for every admissible computational-basis input and error assignment, 
every computational-basis state in the support of the concrete output satisfies the basis-value conditions of $P$.
The tag conditions also hold with the interpretation given by Theorem~\ref{thm:soundness}.

If $P$ fixes a measured ancilla as a function only of the error variables, as is the case for all syndrome properties considered here, the measurement outcome is deterministic and equal to that value.
The measurement therefore does not disturb the encoded superposition.
By linearity, the same result applies to any superposition of admissible basis inputs, including a logical state that is a superposition of code words.
\end{corollary}

\begin{remark}[Scope of Theorem~\ref{thm:soundness} and counterexamples]
\label{rem:cex}
The abstraction is sound on $\mathcal{F}$, but it is not complete.
It is also not sound for circuits outside $\mathcal{F}$.
The following four circuits each violate one condition of Definition~\ref{def:fragment} and show why these conditions are required.
For each circuit, the abstraction gives a different result from the exact state-vector simulation, as verified by the differential test in Section~\ref{sec:ressound}.

\begin{itemize}
\item[C1] $H;X;H$ on $|0\rangle$.
Here, a physical bit flip occurs inside a frame segment.
Since $HXH=Z$, the exact output is $|0\rangle$, whereas the abstraction gives $q_0=1$.
This is the composability restriction discussed in Section~\ref{sec:baseabs}.

\item[C2] $H(q_0);\text{C-NOT}(q_0,q_1);H(q_0)$.
Here, a second Hadamard gate is applied after $q_0$ has participated in an entangling C-NOT.
The exact output is a uniform superposition of all four basis states, whereas the abstraction gives $q_0=q_1$ with an even tag on $q_0$.

\item[C3] $\text{C-NOT}(q_0,q_1);H(q_0)$ with $q_0$ free.
The Hadamard gate is applied after the qubit has already interacted with another qubit.
In the exact circuit, $q_0$ is uniformly random and independent of $q_1$.
The abstraction instead makes the branch of $q_0$ depend on the value of $q_1$.

\item[C4] $H(q_0);\text{C-NOT}(q_0,q_1);Z(q_0);H(q_0)$.
Here, a C-NOT acts inside a frame segment.
\end{itemize}

The abstraction is incomplete for two main reasons.
First, it does not represent the relative phases $\omega$ in Eq.~\eqref{eq:soundness}.
Consequently, states such as
$(|00\rangle+|11\rangle)/\sqrt2$
and
$(|00\rangle-|11\rangle)/\sqrt2$
are represented in the same way.
Second, the tags count Hadamard gates.
Therefore, even a functionally redundant pair of consecutive Hadamard gates changes the tag and can be reported as a violation by a property requiring a particular tag value.
This is a conservative false alarm rather than a missed fault.

Supporting phase-sensitive gates such as $S$, and therefore genuine $Y$-type stabilizers, would require the relative phase information $\omega$ to be represented explicitly, as discussed in Section~\ref{sec8}.
\end{remark}

\subsection{Proofs of Lemmas 1-6}
\label{sec:proofs}
The preliminary version of this work~\cite{bib34} stated six lemmas connecting circuit modifications with violations of the corresponding correctness properties, but did not provide their proofs.
The complete proofs are given below.
Lemmas~\ref{lem:3}--\ref{lem:6} are also stated more precisely than in~\cite{bib34}.
A modification that does not change the function computed by the circuit, such as swapping the order of two C-NOT gates acting on the same ancilla or re-wiring a GHZ C-NOT to a control that already carries the GHZ value, is not a functional fault and cannot be detected by a functional property.
The lemmas therefore characterise which modifications are detected by the proposed properties.

All lemmas consider circuits in $\mathcal{F}$ and are proved in the abstract domain.
Corollary~\ref{cor:props} then transfers the results to the corresponding concrete circuits.

We use the following observation throughout the proofs.
In every circuit family considered here, the basis bit of a qubit at any point in a circuit containing only $X$ and C-NOT gates is an \emph{affine} function over GF(2) of the free input and error variables.
Hadamard gates do not change $q$.
For ancilla $j$, this function can be written as
\[
f_j(\cdot)=\alpha_j\cdot d\oplus\gamma_j,
\]
where $d$ is the vector of data basis values after error injection.
For the repetition code, these variables are $c\oplus e_i$, while for CSS codes they are $v\oplus\mathbf{e}$.
The vector $\alpha_j\in\{0,1\}^N$ is the \emph{parity set} of ancilla $j$, namely the set of data qubits whose values reach the ancilla an odd number of times.
The value $\gamma_j$ is the initial value of the ancilla.
In the intended circuit, $\alpha_j$ is the $j$th row of the relevant check matrix and $\gamma_j=0$.

\begin{lemma}\label{lem:1}
If an ancilla is incorrectly initialised or one or more Hadamard gates are applied to an ancilla, then Properties~\ref{prop:1},~\ref{prop:2},~\ref{prop:genrep},~\ref{prop:genshor}, and~\ref{prop:css} are violated.
\end{lemma}

\begin{proof}
An incorrect initialisation $q^{in}_{Aj}=b'01$ gives $\gamma_j=1$.
For the error-free input, which is allowed by every property, ancilla $j$ therefore has value $1$ instead of $0$.
This violates the syndrome condition.

For an applied Hadamard gate, the properties require an ancilla initialised with $s^{in}_{Aj}=b'000$ to retain $s^{out}_{Aj}=b'000$.
By Definition~\ref{def:h}, each Hadamard gate increases the tag, and saturation at $b'111$ never returns it to $b'000$.
Thus, any number $h\geq1$ of Hadamard gates gives
$s^{out}_{Aj}\neq b'000$ and violates clause~1.
In particular, two consecutive Hadamard gates produce
$s^{out}_{Aj}=b'010$.

The pair is functionally redundant because $H^2=I$.
It is therefore an example of the conservative behaviour described in Remark~\ref{rem:cex}: the modification is reported even though the circuit still computes the intended syndrome.
\end{proof}

\begin{lemma}\label{lem:2}
A missing or additional Hadamard gate on a logical qubit, whether a data or block qubit, 
violates Properties~\ref{prop:1},~\ref{prop:2},~\ref{prop:genrep},~\ref{prop:genshor}, and~\ref{prop:css}.
\end{lemma}

\begin{proof}
The phase-flip properties require each logical qubit to follow the sequence
\[
s=b'000\rightarrow b'001\rightarrow b'010,
\]
which corresponds to exactly two Hadamard operations.
Since the tag records the number of Hadamard gates (Definition~\ref{def:h}), removing either gate or adding another one changes the final tag from $b'010$ and therefore violates clause~1.

For the bit-flip properties, logical qubits must not receive any Hadamard gate.
Consequently, applying an $H$ gate changes the required output tag from $s^{out}_{Li}=b'000$ and violates the property.
\end{proof}

\begin{lemma}\label{lem:3}
For bit-flip detection, a modification to the C-NOT gates in the syndrome-extraction circuit, including a missing, additional, re-wired, or control/target-exchanged C-NOT, violates Properties~\ref{prop:1},~\ref{prop:genrep}, and~\ref{prop:css} if and only if it changes the parity set $\alpha_j$ of at least one ancilla.
\end{lemma}

\begin{proof}
If all parity sets remain unchanged, every ancilla computes the intended function.
Since the properties constrain the ancilla values and tags, they continue to hold.

Conversely, suppose that
$\alpha_j\neq\alpha_j^{*}$ for some ancilla $j$, where $\alpha_j^{*}$ is the intended parity set.
Choose a qubit
\[
i\in\alpha_j\,\triangle\,\alpha_j^{*}
\]
from their symmetric difference.
Consider the admissible input consisting of the error-free code word $0$, that is, $c=0$ for the repetition code or $v=0$ for the CSS code, together with a single error on qubit $i$.
The resulting data vector $d$ is the $i$th unit vector.
The value computed by ancilla $j$ then differs from the intended value $\alpha_j^{*}\cdot d$ by one, so the syndrome condition is violated.

A missing C-NOT removes a qubit from the parity set.
An additional C-NOT can change the parity set, and a re-wired C-NOT replaces one qubit by another.
Exchanging the control and target also removes the data qubit from the parity set.
In addition, such an exchange can corrupt the target data qubit, which is detected by the data-preservation condition of Property~\ref{prop:anyerr}.

By contrast, reordering the C-NOT gates acting on the same ancilla does not change its parity set.
It is therefore not detected, since the reordered circuit is functionally equivalent to the intended circuit.
\end{proof}

\begin{lemma}\label{lem:4}
The statement of Lemma~\ref{lem:3} also holds for the phase-flip detection circuits and Properties~\ref{prop:2},~\ref{prop:genshor}, and~\ref{prop:css}.
\end{lemma}

\begin{proof}
In the phase-flip circuits, the ancillas are computed after the second Hadamard gate of each data or block qubit.
At this point, the tags are even and $q_0$ represents the computational-basis value
$v_i\oplus e_i$ for CSS codes, or $c\oplus e_i$ for the concatenated code (Theorem~\ref{thm:soundness}, case~6).
The argument of Lemma~\ref{lem:3} therefore applies directly.
For Property~\ref{prop:genshor}, the parity sets are defined over logical blocks rather than individual data qubits.
\end{proof}

\begin{lemma}\label{lem:5}
A modification to a Bell-state or GHZ-state generation circuit, including a missing, additional, or incorrectly placed Hadamard gate, or a missing, additional, or incorrectly connected C-NOT gate, violates Property~\ref{prop:3} if and only if it changes the Hadamard count of a qubit or changes the function computed by the basis bit of a qubit.
In the intended circuit, the latter function is
$q_0^{out}=x_1$ for every qubit.
\end{lemma}

\begin{proof}
Clauses~1 and~2 of Property~\ref{prop:3} fix the required Hadamard count for every qubit.
Therefore, any change in these counts violates the property.

The only free variable is the input value $x_1$ of the first qubit.
Every output basis bit is consequently an affine function of $x_1$, namely one of
$0$, $1$, $x_1$, or $x_1\oplus1$.
Clause~3 requires all output functions to equal the function of the first qubit, which is $x_1$ unless a C-NOT targets the first qubit.

If a qubit computes a different affine function, that function differs from $x_1$ for at least one value of $x_1$, and the solver produces a counterexample.
A missing C-NOT leaves its target at $0$.
Exchanging the control and target also leaves the target at $0$ because the control qubit initially has value $0$ and does not change the first qubit.
A re-targeted C-NOT can leave its intended target at $0$ or flip another qubit twice.
An additional C-NOT from a qubit carrying $x_1$ to another qubit carrying $x_1$ resets that qubit to $0$.
These cases therefore violate clause~3.

In contrast, re-wiring a C-NOT to another control that already carries $x_1$ leaves all computed functions unchanged.
Such a circuit is still a valid GHZ circuit.
The three topologies considered in Section~\ref{sec:resghz} differ in this way, so the modification is correctly not reported.
\end{proof}

\begin{lemma}\label{lem:6}
Let a circuit be obtained from an intended circuit by any combination of the modifications considered in Lemmas~\ref{lem:1}--\ref{lem:5}.
The resulting circuit violates the corresponding property among
Properties~\ref{prop:1}--\ref{prop:3},
\ref{prop:genrep},
\ref{prop:genshor},
\ref{prop:steane}, and
\ref{prop:css}
if and only if the combination changes the Hadamard count of some qubit, the initial value of an ancilla, or the affine function computed by an observed qubit.
The latter includes the parity set of an ancilla and the basis bit of a qubit in a Bell or GHZ circuit.
\end{lemma}

\begin{proof}
The conditions in Lemmas~\ref{lem:1}-\ref{lem:5} are expressed in terms of the Hadamard counts, initial values, and affine functions of the resulting circuit, rather than in terms of individual modifications.
The same arguments therefore apply after any combination of the considered modifications.

For the ``if'' direction, each changed quantity gives an admissible counterexample, as established by the preceding lemmas.
For the ``only if'' direction, the properties constrain only these quantities.

Consequently, two modifications can cancel each other.
For example, two identical additional C-NOT gates acting on the same ancilla leave its parity unchanged.
The resulting circuit is therefore functionally equivalent to the intended circuit and is not reported.

The SMT solver checks each property for all admissible assignments.
Hence, a single counterexample is sufficient to reject a circuit.
The mutation experiments in Section~\ref{sec:resmutation} provide empirical confirmation of this characterisation for the implemented encoding.
\end{proof}

\section{Verification Methodology}\label{sec:methodology}

This section describes the SMT-based verification workflow, the size of the generated constraint system, 
the mutation-based validation strategy used to evaluate the proposed properties, and the experimental environment in which all evaluations were conducted.

\subsection{SMT encoding and verification}\label{sec:encoding}

Each circuit family introduced in Section~\ref{sec:absfrmwrk} is encoded directly using the Z3 Python API~\cite{bib33}. This avoids the intermediate SMT-LIB representation used in our earlier implementation. 
Each abstract qubit is represented by two fixed-width \texttt{BitVec} terms for its $(s,q)$ components. 
The abstract gates defined in Definitions~\ref{def:h}-\ref{def:cnot} are encoded as equality constraints between fresh output bit-vectors and the corresponding input bit-vectors. 
Conditional behaviour is represented using Z3's \texttt{If} expressions. 
Constant inputs, such as ancillas initialised to $b'00$, are represented by bit-vector literals. 
Similarly, the tags of qubits that do not receive a Hadamard gate remain literals throughout the encoding.

Faults and error locations are represented by free 1-bit variables. 
For properties that restrict the number of errors, these variables are constrained by a cardinality condition $\sum_i e_i\le w$, 
encoded using Z3's native pseudo-Boolean constraint \texttt{AtMost}. 
Code-word preconditions ($\mathbf{H}_Zv=0$ or $\mathbf{H}_Xv=0$) are encoded as XOR constraints over the free data variables. 
Given a correctness property $P$, 
verification asserts the circuit constraints together with its negation, 
$\neg P$, and invokes the solver using \texttt{check()}. 
An \texttt{unsat} result shows that the property holds for all valid assignments of the symbolic variables. 
For circuits in $\mathcal{F}$, Corollary~\ref{cor:props} then transfers this result to the corresponding concrete circuit. 
Conversely, a \texttt{sat} result provides a counterexample, namely an assignment that violates the property.

\paragraph{Size of the constraint system.}
Consider a circuit with $n$ qubits, $g_H$ Hadamard gates, 
$g_X$ X gates or error injections, and $g_C$ C-NOT gates. 
The encoding introduces one fresh 3-bit variable for each Hadamard gate and one fresh 2-bit variable 
for each X gate, error injection, C-NOT gate and measurement, 
together with one equality constraint for each. Thus, the number of SMT variables and assertions is $O(n+g_H+g_X+g_C)$. 
It is independent of the circuit depth except through the number of gates. 
Each constraint is a fixed-width \texttt{If}/XOR relation 
involving at most three bit-vectors of width at most three. 
After bit-blasting, this gives a constant number of clauses per gate, 
so the propositional encoding also has $O(n+g_H+g_X+g_C)$ clauses.

The property itself is the main source of non-constant-size constraints, 
with one conjunct for each checked qubit. 
The cardinality constraint over the $N$ error variables is handled by Z3's pseudo-Boolean reasoning; 
for the clause counts reported below, the \texttt{card2bv} tactic compiles this constraint into $O(N\log N)$ clauses. 
The code-word preconditions contribute $O(\text{nnz}(\mathbf{H}))$ XOR terms. 
For a GHZ circuit with $m$ qubits, 
the encoding therefore contains $m+1$ variables and $m$ gate constraints for all three topologies. 
For a CSS circuit, the size is $O(N+\text{nnz}(\mathbf{H}))$.

Section~\ref{sec:ressize} reports the measured numbers of variables, 
assertions and CNF clauses after Z3's \texttt{simplify}, \texttt{card2bv}, \texttt{bit-blast} and \texttt{tseitin-}\allowbreak\texttt{cnf} tactics. 
These measurements show the expected linear growth of the encoding. 
Linear encoding size, however, does not imply linear solving time. 
In general, deciding the validity of such formulas is co-NP-complete, 
and solver performance depends on the structure of the constraints. 
Section~\ref{sec:resrandom} shows that dense XOR dependencies, 
such as those produced by deep random C-NOT networks, are considerably harder to solve than the sparse dependencies in the code and GHZ circuits. 
Cardinality constraints provide a second source of difficulty. 
For large surface codes, the single-error query is slower than the arbitrary-error query of Property~\ref{prop:anyerr}, 
even though the latter implies it. 
Decoders that compare error weights also make recovery queries more expensive (Section~\ref{sec:resmulti}).

\subsection{Mutation-based fault injection}\label{sec:mutationmethod}
To complement the formal proofs in Section~\ref{sec:proofs}, 
we evaluate the fault-detection capability of the proposed framework using mutation testing. 
Each circuit-construction routine accepts an optional mutation operator and its target location, 
allowing faulty circuit variants to be generated automatically. 
The operators, their target locations and the circuit families to which they apply are listed in Table~\ref{tab:mutops}. 
They cover the fault classes considered in Lemmas~\ref{lem:1}-\ref{lem:5}, 
including missing, additional and incorrectly connected gates, as well as incorrect ancilla initialisation.

Each operator is applied at every valid location in every circuit instance. 
The resulting distribution of mutants is therefore determined by the circuit structure rather than selected manually. 
This avoids selection bias and makes the complete sweep exactly reproducible.

\begin{table}[!h]
\centering
\caption{Mutation operators. ``Ancilla'' operators are applied to every ancilla (parity check) of every instance, ``data'' operators to every data or block qubit, and GHZ operators at positions $1$, $m/2$ and $m-1$.}
\label{tab:mutops}
\footnotesize
\begin{tabularx}{\textwidth}{
>{\raggedright\arraybackslash}p{1.95cm}
>{\raggedright\arraybackslash}X
>{\raggedright\arraybackslash}X
}
\toprule
Operator & Modification & Applies to \\
\midrule
\texttt{drop\_cnot} & omit the last C-NOT of an ancilla's parity computation & ancilla \\
\texttt{wrong\_wire} & re-wire the last C-NOT of an ancilla to a data qubit outside the check & ancilla \\
\texttt{extra\_cnot} & add a C-NOT from a data qubit outside the check & ancilla \\
\texttt{swapped\_cnot} & exchange control and target of the last C-NOT of an ancilla & ancilla \\
\texttt{reorder\_cnot} & reverse the order of an ancilla's C-NOT gates (equivalent by construction) & ancilla \\
\texttt{mis\_init} & initialise the ancilla to $q=b'01$ instead of $b'00$ & ancilla \\
\texttt{extra\_h\_anc} & add a Hadamard gate on the ancilla before measurement & ancilla \\
\texttt{missing\_h1}, \texttt{missing\_h2} & omit the first or second Hadamard gate of a data/block qubit & data (phase-flip) \\
\texttt{extra\_h} & add a third Hadamard gate on a data/block qubit & data (phase-flip) \\
\texttt{missing\_h}, \texttt{extra\_h} & omit, or duplicate, the Hadamard gate of a Bell/GHZ circuit & Bell, GHZ \\
\texttt{missing\_cnot} & omit the C-NOT that entangles a given qubit & Bell, GHZ \\
\texttt{swapped\_cnot} & exchange control and target of that C-NOT & Bell, GHZ \\
\texttt{wrong\_target} & apply that C-NOT to the next qubit instead & GHZ \\
\bottomrule
\end{tabularx}
\end{table}
\paragraph{Equivalent and non-generated mutants.}
A mutant that produces exactly the same abstract outputs 
as the original circuit for every admissible input cannot be detected by any functional property. 
Following standard mutation-testing practice~\cite{bib38,bib52}, 
such mutants are excluded.
Rather than identifying them by inspection, we determine equivalence mechanically using an SMT \emph{miter}.
The original circuit and the mutant are constructed using the same input 
and error variables but separate internal variables, 
together with the admissibility constraints. 
The solver then searches for an assignment for which any output component $(s,q)$ of any qubit differs between the two circuits.

An \texttt{unsat} result proves equivalence, whereas a \texttt{sat} result proves non-equivalence. 
Since the abstract circuits are finite-state, this procedure is complete, 
so no equivalent mutant remains unidentified among the evaluated mutants. 
Only mutants proved to be non-equivalent by the miter are evaluated against the correctness property.

If a mutation location does not exist, the mutant is reported as not generated rather than silently replacing it with the unmutated circuit. 
For example, this occurs when the C-NOT of the last GHZ qubit is re-targeted to a non-existent next qubit. 
The equivalent mutant reported in the submitted version of this article was of this type: 
for the 3-qubit repetition code, the re-wiring target of the second ancilla did not exist. 
With the cyclic choice of re-wiring target used here, this mutation exists and is detected.

For every generated mutant, the corresponding correctness property is verified again. 
A mutant is considered \emph{detected} when the solver returns \texttt{sat}, 
showing that the injected fault violates the property. 
The detection rate for a group of mutants is the proportion of non-equivalent mutants that are detected.

\paragraph{Relationship with the formal proofs.}
Lemmas~\ref{lem:1}-\ref{lem:6} establish, in the abstract domain, 
which modifications violate the corresponding properties. 
The mutation study is therefore not an independent validation of soundness, 
which is addressed by Theorem~\ref{thm:soundness} 
and the differential test in Section~\ref{sec:ressound}. 
Instead, it is used to validate the implementation of the properties and circuit encodings against the lemmas.

\subsection{Experimental setup}\label{sec:setup}
All experiments were performed using Z3 version 4.13.4~\cite{bib33}, 
installed as the \texttt{z3-solver}~4.13.4.0 Python package, 
and Python~3.13.5 on a MacBook with an Apple M1 processor 
(4 performance and 4 efficiency cores) and 16~GB of memory, running macOS 27.0. 
For the solver-version ablation in Section~\ref{sec:resablation}, 
Z3~5.1.0 with Python~3.10.11 was used on the same machine. 
Z3 was run with its default configuration through \texttt{Solver()}, 
using a single thread and no custom tactics. 
No timeout was used except for the random networks in Section~\ref{sec:resrandom}, 
for which a 600s timeout was applied. 
All experiments were run sequentially, with one process at a time.

\paragraph{Timing protocol.}
For each benchmark instance, one warm-up run was discarded, 
followed by $R$ independent repetitions. 
We used $R=10$ for all instances except GHZ circuits with more than 
16,384 qubits and surface codes with more than 2,601 data qubits, 
for which $R=5$, surface codes with more than 5,041 data qubits, 
for which $R=3$, 
and the multi-error properties in Section~\ref{sec:resmulti}, for which $R=3$.

Each repetition rebuilds the circuit and creates a fresh solver, 
so no solver state is shared between repetitions. 
We report the median, mean, standard deviation, minimum and maximum. 
\emph{Solve time} is the wall-clock time required for solver creation, 
constraint assertion and \texttt{check()}. 
\emph{Build time}, which covers the Python-side construction of the circuit terms, 
is reported separately because it dominates the end-to-end time for the largest instances.

In the random-network experiment, each configuration was run once for each of three random seeds. 
\emph{Peak memory} is the peak resident-set size (\texttt{getrusage}, 
\texttt{ru\_maxrss}) of a fresh Python process that builds and verifies a single instance. 
We report the median over three processes, together with the baseline memory usage of the interpreter with Z3 loaded.

Unlike the preliminary study~\cite{bib34}, which used a 24-core workstation with 192~GB of memory, 
the experiments reported here were conducted on a commodity laptop. 
This configuration was chosen to assess the practicality of the proposed abstraction under hardware resources representative of a typical research and software-development environment.


\section{Experimental Results}\label{sec:results}
This section evaluates the proposed verification framework experimentally.
We first reproduce the baseline results from the original study, 
followed by experiments on the generalised repetition and concatenated phase-flip code families.
We then evaluate CSS syndrome extraction, including the Steane $[[7,1,3]]$ code, together with the multiple-error, detection and recovery properties.
Finally, we study scalability using GHZ circuits with three topologies, random networks of increasing depth, 
measurements of encoding size, and a controlled ablation of the encoding pipeline.
We also report the differential soundness test and the mutation study.
The measurement protocol is described in Section~\ref{sec:setup}.

\subsection{Baseline reproduction}
As an initial validation, we reproduced the verification experiments reported in~\cite{bib34} using the implementation described in Section~\ref{sec:methodology}.
Although the experiments were run on different hardware and use direct Z3 API encoding rather than SMT-LIB generation, all benchmark circuits were successfully verified.
Table~\ref{tab:baseline} summarises the results.
Each circuit is verified in about one millisecond, 
with standard deviations of a few percent of the median.
The first query included the one-off initialisation time of the solver library,
and is excluded here by using a warm-up run.

\begin{table}[!h]
\centering
\caption{Baseline verification results for the fixed-size benchmark circuits (10 repetitions after one warm-up run).}
\label{tab:baseline}
\footnotesize
\begin{tabular}{@{}lccrrr@{}}
\toprule
Circuit & Qubits & Verified & Median (ms) & Mean $\pm$ s.d. (ms) & Min--max (ms) \\
\midrule
2-qubit bit-flip & 3 & Yes & 0.809 & 0.824 $\pm$ 0.041 & 0.789--0.923 \\
3-qubit bit-flip & 5 & Yes & 1.04 & 1.05 $\pm$ 0.038 & 0.984--1.11 \\
9-qubit Shor phase-flip & 11 & Yes & 1.18 & 1.18 $\pm$ 0.031 & 1.11--1.22 \\
Bell-state generation & 2 & Yes & 0.398 & 0.400 $\pm$ 0.009 & 0.387--0.417 \\
\bottomrule
\end{tabular}

\end{table}

\FloatBarrier
\subsection{Generalised repetition and concatenated phase-flip codes}
The generalised repetition code (Property~\ref{prop:genrep}) and the generalised concatenated phase-flip code (Property~\ref{prop:genshor}) 
were verified for increasing values of the size parameters $n$ and $k$, 
from the original fixed-size instances up to codes containing 101 logical elements.
The results are given in Table~\ref{tab:scaling}, and Figure~\ref{fig:codescaling} shows the corresponding runtime trends.
All instances satisfy their respective correctness property.
Verification time increases approximately linearly with code size, 
which is consistent with the linear number of parity constraints.
The largest instances complete in 32.8ms and 38.0ms, respectively.

The concatenated-code results use the corrected encoding described in Section~\ref{sec:mutationmethod}.
Its ancilla parities are computed using $2(k-1)$ C-NOT gates.
It therefore has the same order of cost as the repetition code, 
rather than the nearly constant cost obtained with the submitted implementation.

\begin{table}[!h]
\centering
\caption{Solve time for the generalised repetition and concatenated phase-flip codes (median $\pm$ standard deviation over 10 repetitions). The concatenated code is modelled with one abstract block-qubit per three-qubit block; the physical-equivalent count is $3k+(k-1)$.}
\label{tab:scaling}
\footnotesize
\begin{tabular}{@{}rrrrrr@{}}
\toprule
 & \multicolumn{2}{c}{Repetition code ($n$)} & \multicolumn{3}{c}{Concatenated phase-flip code ($k$)} \\
\cmidrule(lr){2-3}\cmidrule(l){4-6}
$n$ / $k$ & Qubits & Median $\pm$ s.d. (ms) & Abstract qubits & Physical equiv. & Median $\pm$ s.d. (ms) \\
\midrule
2 & 3 & 0.763 $\pm$ 0.025 & -- & -- & -- \\
3 & 5 & 1.02 $\pm$ 0.041 & 5 & 11 & 1.15 $\pm$ 0.063 \\
4 & 7 & 1.21 $\pm$ 0.049 & -- & -- & -- \\
5 & 9 & 1.41 $\pm$ 0.049 & 9 & 19 & 1.63 $\pm$ 0.028 \\
7 & 13 & 1.80 $\pm$ 0.074 & 13 & 27 & 2.19 $\pm$ 0.095 \\
9 & 17 & 2.24 $\pm$ 0.048 & 17 & 35 & 2.76 $\pm$ 0.053 \\
11 & 21 & 2.59 $\pm$ 0.077 & 21 & 43 & 3.22 $\pm$ 0.044 \\
15 & 29 & 3.39 $\pm$ 0.093 & 29 & 59 & 4.08 $\pm$ 0.094 \\
21 & 41 & 4.89 $\pm$ 0.152 & 41 & 83 & 5.60 $\pm$ 0.091 \\
31 & 61 & 7.57 $\pm$ 0.177 & 61 & 123 & 8.66 $\pm$ 0.284 \\
41 & 81 & 10.4 $\pm$ 0.468 & 81 & 163 & 11.9 $\pm$ 0.316 \\
51 & 101 & 13.3 $\pm$ 0.351 & 101 & 203 & 15.8 $\pm$ 0.488 \\
63 & 125 & 17.5 $\pm$ 0.608 & 125 & 251 & 19.5 $\pm$ 0.665 \\
81 & 161 & 24.7 $\pm$ 0.866 & 161 & 323 & 27.6 $\pm$ 0.788 \\
101 & 201 & 32.8 $\pm$ 1.27 & 201 & 403 & 38.0 $\pm$ 1.57 \\
\bottomrule
\end{tabular}

\end{table}

\begin{figure}[!h]
\centering
\includegraphics[width=0.8\textwidth]{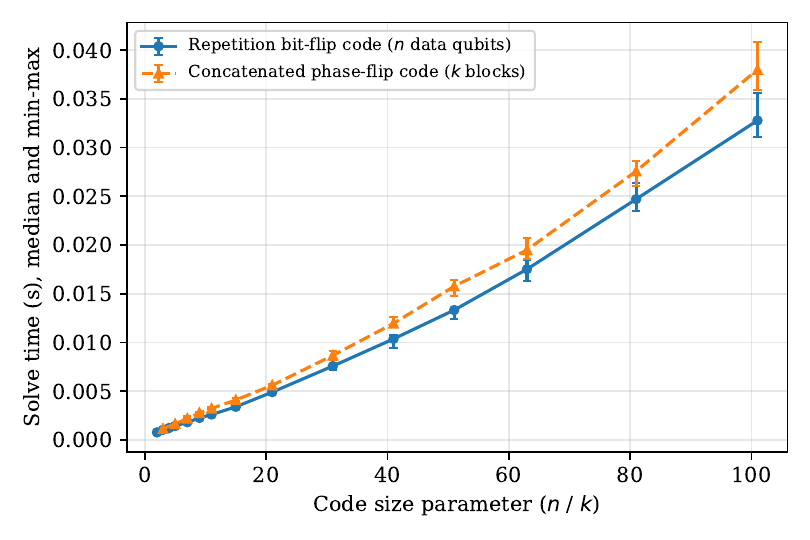}
\caption{Solve time as a function of code size for the generalised repetition and concatenated phase-flip codes (median; bars show minimum and maximum over 10 repetitions).}
\label{fig:codescaling}
\end{figure}

\FloatBarrier
\subsection{CSS syndrome extraction}\label{sec:rescss}
To test whether the abstraction also applies beyond repetition-based constructions,
we verified bit-flip and phase-flip syndrome extraction (Definition~\ref{def:css})
for four CSS code families with different parity-check structures:
the Steane code ($\mathbf{H}_X=\mathbf{H}_Z$, weight-4 checks);
the $[[15,7,3]]$ quantum Hamming code ($\mathbf{H}_X=\mathbf{H}_Z$, weight-8 checks);
the $[[15,1,3]]$ quantum Reed-Muller code ($\mathbf{H}_X\neq\mathbf{H}_Z$, 
with 4 $X$-type and 10 $Z$-type checks of weights 8 and 4, and distances $d_X=7$, $d_Z=3$);
and rotated surface codes ($\mathbf{H}_X\neq\mathbf{H}_Z$, 
with local weight-2 and weight-4 checks on a two-dimensional lattice and degenerate syndromes).

The matrices are generated programmatically.
The experiment checks the CSS condition~\eqref{eq:csscond}, 
the number of logical qubits and, for codes with $N\le25$, the distances by exhaustive enumeration before verification.

Tables~\ref{tab:steane} and~\ref{tab:surface} report the results for the single-error property (Properties~\ref{prop:steane} and~\ref{prop:css}) 
and the arbitrary-error property (Property~\ref{prop:anyerr}).
All syndrome-extraction circuits satisfy their properties.
The Steane circuits are verified in about 4.50ms.
The surface-code family was verified up to distance 101, corresponding to 15,301 data and ancilla qubits, in 160s (Figure~\ref{fig:surface}).
For surface codes with $d>11$, only the arbitrary-error property was verified.
This property implies the single-error property, while its cardinality constraint makes the latter slower for large $d$ (Section~\ref{sec:encoding}).

For this family, solve time grows faster than linearly with the number of qubits, from 0.278s at $d=21$ to 160s at $d=101$.
These results show that the abstraction can represent parity-check computations defined by arbitrary CSS check matrices, 
including cases with unequal $\mathbf{H}_X$ and $\mathbf{H}_Z$ and local two-dimensional check structures, 
without changing the underlying gate abstraction.

\begin{table}[!h]
\centering
\caption{Verification of CSS syndrome extraction for the Steane, quantum Hamming and quantum Reed-Muller codes. ``Checks'' is the number of rows of $\mathbf{H}_Z$ (bit-flip) or $\mathbf{H}_X$ (phase-flip). ``Arbitrary'' refers to Property~\ref{prop:anyerr}.}
\label{tab:steane}
\footnotesize
\begin{tabular}{@{}llrrlccr@{}}
\toprule
Code & Frame & $N$ & Checks & Errors & Verified & Qubits & Median $\pm$ s.d. (ms) \\
\midrule
Steane [[7,1,3]] & bit & 7 & 3 & $\le 1$ error & Yes & 10 & 4.39 $\pm$ 0.357 \\
Steane [[7,1,3]] & bit & 7 & 3 & arbitrary & Yes & 10 & 5.29 $\pm$ 0.757 \\
Steane [[7,1,3]] & phase & 7 & 3 & $\le 1$ error & Yes & 10 & 4.61 $\pm$ 0.369 \\
Steane [[7,1,3]] & phase & 7 & 3 & arbitrary & Yes & 10 & 4.38 $\pm$ 1.19 \\
Hamming [[15,7,3]] & bit & 15 & 4 & $\le 1$ error & Yes & 19 & 75.7 $\pm$ 26.0 \\
Hamming [[15,7,3]] & bit & 15 & 4 & arbitrary & Yes & 19 & 81.2 $\pm$ 53.7 \\
Hamming [[15,7,3]] & phase & 15 & 4 & $\le 1$ error & Yes & 19 & 58.2 $\pm$ 13.9 \\
Hamming [[15,7,3]] & phase & 15 & 4 & arbitrary & Yes & 19 & 84.0 $\pm$ 85.0 \\
Reed-Muller [[15,1,3]] & bit & 15 & 10 & $\le 1$ error & Yes & 25 & 30.1 $\pm$ 2.09 \\
Reed-Muller [[15,1,3]] & bit & 15 & 10 & arbitrary & Yes & 25 & 143 $\pm$ 60.2 \\
Reed-Muller [[15,1,3]] & phase & 15 & 4 & $\le 1$ error & Yes & 19 & 70.0 $\pm$ 21.4 \\
Reed-Muller [[15,1,3]] & phase & 15 & 4 & arbitrary & Yes & 19 & 59.6 $\pm$ 30.2 \\
\bottomrule
\end{tabular}

\end{table}

\begin{table}[!h]
\centering
\caption{Verification of syndrome extraction for rotated surface codes of distance $d$ ($d^2$ data qubits and $(d^2-1)/2$ ancillas for the checks of one frame). All instances are verified. ``Reps'' is the number of repetitions.}
\label{tab:surface}
\footnotesize
\begin{tabular}{@{}rrrrrrrc@{}}
\toprule
 & & & \multicolumn{2}{c}{$\le1$ error, median (s)} & \multicolumn{2}{c}{Arbitrary errors, median (s)} & \\
\cmidrule(lr){4-5}\cmidrule(lr){6-7}
$d$ & Data & Qubits & bit & phase & bit & phase & Reps \\
\midrule
3 & 9 & 13 & 0.0037 & 0.0044 & 0.0028 & 0.0032 & 10 \\
5 & 25 & 37 & 0.0162 & 0.0189 & 0.0078 & 0.0092 & 10 \\
7 & 49 & 73 & 0.0422 & 0.0464 & 0.0149 & 0.0165 & 10 \\
9 & 81 & 121 & 0.0998 & 0.104 & 0.0256 & 0.0285 & 10 \\
11 & 121 & 181 & 0.214 & 0.218 & 0.0415 & 0.0462 & 10 \\
15 & 225 & 337 & -- & -- & 0.0941 & 0.104 & 10 \\
21 & 441 & 661 & -- & -- & 0.260 & 0.278 & 10 \\
31 & 961 & 1{,}441 & -- & -- & 1.28 & 1.25 & 10 \\
51 & 2{,}601 & 3{,}901 & -- & -- & 9.51 & 10.5 & 10 \\
71 & 5{,}041 & 7{,}561 & -- & -- & 35.1 & 29.6 & 5 \\
101 & 10{,}201 & 15{,}301 & -- & -- & 160 & 154 & 3 \\
\bottomrule
\end{tabular}

\end{table}

\begin{figure}[!h]
\centering
\includegraphics[width=0.8\textwidth]{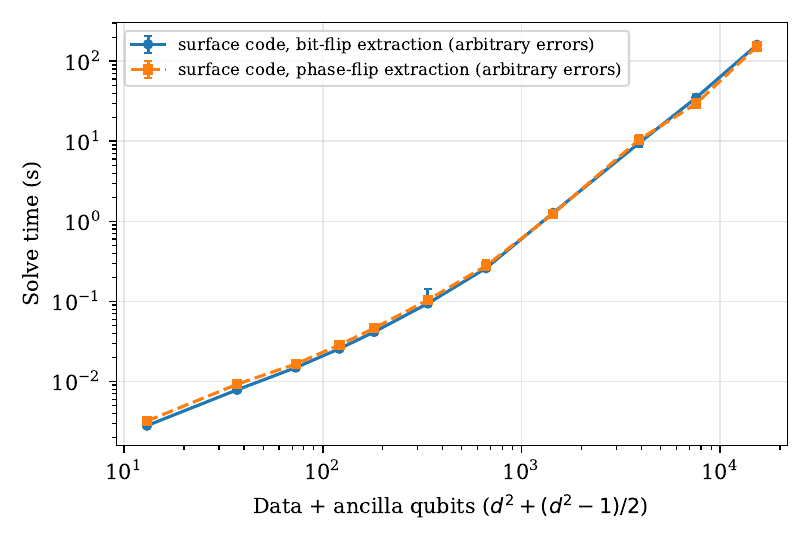}
\caption{Solve time of the arbitrary-error syndrome property for rotated surface codes (median; bars show minimum and maximum).}
\label{fig:surface}
\end{figure}

\FloatBarrier
\subsection{Multiple errors, detection, and recovery}
\label{sec:resmulti}
Table~\ref{tab:multi} reports Properties~\ref{prop:anyerr}-\ref{prop:recover} for the repetition codes and the three small CSS codes.
For each property, the table includes both the largest error weight for which the property is expected to hold and the next weight.
All 60 outcomes agree with coding theory.

Syndrome extraction is correct for every error pattern.
Detection holds up to weight $d-1$ ($n-1$ for the repetition code), 
while at weight $d$ the solver returns an undetectable error pattern.
For the Steane code, for example, this is a weight-3 pattern that corresponds to a logical operator.

Recovery holds up to weight $\lfloor(d-1)/2\rfloor$.
At the next weight, the solver returns an error pattern that the decoder miscorrects.
For example, for the Steane code, two bit flips can produce a syndrome equal to the column associated with a third qubit.
The quantum Reed-Muller code shows the expected asymmetry: bit-flip errors are detected up to weight~6, 
whereas phase-flip errors are detected only up to weight~2.

Recovery queries for the repetition code are substantially more expensive than the other queries
because the minimum-weight decoder compares error weights.
For $n=21$, the solve time is 69.2s.
In a preliminary run, $n=31$ did not finish within 120s.
Thus, verifying weight-based decoders at larger scales remains a limitation of the current encoding.
Look-up and syndrome-only properties remain inexpensive.

\begin{table}[!h]
\centering
\caption{Multi-error properties (median of 3 repetitions after a warm-up run). ``Violated'' rows are the expected failures beyond the capability of the code. The counterexample column lists the erroneous qubits of the pattern returned by the solver.}
\label{tab:multi}
\scriptsize
\begin{tabular}{@{}llrcr>{\raggedright\arraybackslash}p{0.2\textwidth}@{}}
\toprule
Circuit & Property & $w$ & Result & Median (ms) & Counterexample \\
\midrule
Repetition, $n=3$ & Prop.~\ref{prop:anyerr} & any & verified & 9.47 &  \\
Repetition, $n=3$ & Prop.~\ref{prop:detect} & 2 & verified & 7.64 &  \\
Repetition, $n=3$ & Prop.~\ref{prop:detect} & 3 & violated & 8.60 & $\{0,1,2\}$ \\
Repetition, $n=3$ & Prop.~\ref{prop:recover} & 1 & verified & 11.1 &  \\
Repetition, $n=3$ & Prop.~\ref{prop:recover} & 2 & violated & 10.5 & $\{1,2\}$ \\
Repetition, $n=5$ & Prop.~\ref{prop:anyerr} & any & verified & 10.8 &  \\
Repetition, $n=5$ & Prop.~\ref{prop:detect} & 4 & verified & 9.80 &  \\
Repetition, $n=5$ & Prop.~\ref{prop:detect} & 5 & violated & 9.64 & $\{0,\dots,4\}$ \\
Repetition, $n=5$ & Prop.~\ref{prop:recover} & 2 & verified & 12.1 &  \\
Repetition, $n=5$ & Prop.~\ref{prop:recover} & 3 & violated & 11.2 & $\{0,1,2\}$ \\
Repetition, $n=7$ & Prop.~\ref{prop:anyerr} & any & verified & 11.3 &  \\
Repetition, $n=7$ & Prop.~\ref{prop:detect} & 6 & verified & 9.22 &  \\
Repetition, $n=7$ & Prop.~\ref{prop:detect} & 7 & violated & 9.75 & $\{0,\dots,6\}$ \\
Repetition, $n=7$ & Prop.~\ref{prop:recover} & 3 & verified & 19.3 &  \\
Repetition, $n=7$ & Prop.~\ref{prop:recover} & 4 & violated & 13.3 & $\{0,1,3,4\}$ \\
Repetition, $n=9$ & Prop.~\ref{prop:anyerr} & any & verified & 11.2 &  \\
Repetition, $n=9$ & Prop.~\ref{prop:detect} & 8 & verified & 10.2 &  \\
Repetition, $n=9$ & Prop.~\ref{prop:detect} & 9 & violated & 10.3 & $\{0,\dots,8\}$ \\
Repetition, $n=9$ & Prop.~\ref{prop:recover} & 4 & verified & 47.4 &  \\
Repetition, $n=9$ & Prop.~\ref{prop:recover} & 5 & violated & 11.1 & $\{0,\dots,4\}$ \\
Repetition, $n=15$ & Prop.~\ref{prop:anyerr} & any & verified & 10.7 &  \\
Repetition, $n=15$ & Prop.~\ref{prop:detect} & 14 & verified & 10.7 &  \\
Repetition, $n=15$ & Prop.~\ref{prop:detect} & 15 & violated & 10.8 & $\{0,\dots,14\}$ \\
Repetition, $n=15$ & Prop.~\ref{prop:recover} & 7 & verified & 1,546 &  \\
Repetition, $n=15$ & Prop.~\ref{prop:recover} & 8 & violated & 16.0 & $\{0,\dots,7\}$ \\
Repetition, $n=21$ & Prop.~\ref{prop:anyerr} & any & verified & 12.6 &  \\
Repetition, $n=21$ & Prop.~\ref{prop:detect} & 20 & verified & 11.5 &  \\
Repetition, $n=21$ & Prop.~\ref{prop:detect} & 21 & violated & 11.6 & $\{0,\dots,20\}$ \\
Repetition, $n=21$ & Prop.~\ref{prop:recover} & 10 & verified & 69,204 &  \\
Repetition, $n=21$ & Prop.~\ref{prop:recover} & 11 & violated & 16.4 & $\{0,\dots,10\}$ \\
Steane, bit & Prop.~\ref{prop:anyerr} & any & verified & 11.4 &  \\
Steane, bit & Prop.~\ref{prop:detect} & 2 & verified & 14.9 &  \\
Steane, bit & Prop.~\ref{prop:detect} & 3 & violated & 14.2 & $\{1,3,5\}$ \\
Steane, bit & Prop.~\ref{prop:recover} & 1 & verified & 15.9 &  \\
Steane, bit & Prop.~\ref{prop:recover} & 2 & violated & 12.4 & $\{2,5\}$ \\
Steane, phase & Prop.~\ref{prop:anyerr} & any & verified & 13.9 &  \\
Steane, phase & Prop.~\ref{prop:detect} & 2 & verified & 17.6 &  \\
Steane, phase & Prop.~\ref{prop:detect} & 3 & violated & 14.9 & $\{1,3,5\}$ \\
Steane, phase & Prop.~\ref{prop:recover} & 1 & verified & 18.7 &  \\
Steane, phase & Prop.~\ref{prop:recover} & 2 & violated & 13.4 & $\{1,2\}$ \\
{}[[15,7,3]], bit & Prop.~\ref{prop:anyerr} & any & verified & 72.1 &  \\
{}[[15,7,3]], bit & Prop.~\ref{prop:detect} & 2 & verified & 80.3 &  \\
{}[[15,7,3]], bit & Prop.~\ref{prop:detect} & 3 & violated & 15.5 & $\{4,7,12\}$ \\
{}[[15,7,3]], bit & Prop.~\ref{prop:recover} & 1 & verified & 111 &  \\
{}[[15,7,3]], bit & Prop.~\ref{prop:recover} & 2 & violated & 14.5 & $\{12,14\}$ \\
{}[[15,7,3]], phase & Prop.~\ref{prop:anyerr} & any & verified & 117 &  \\
{}[[15,7,3]], phase & Prop.~\ref{prop:detect} & 2 & verified & 79.0 &  \\
{}[[15,7,3]], phase & Prop.~\ref{prop:detect} & 3 & violated & 17.9 & $\{3,7,11\}$ \\
{}[[15,7,3]], phase & Prop.~\ref{prop:recover} & 1 & verified & 144 &  \\
{}[[15,7,3]], phase & Prop.~\ref{prop:recover} & 2 & violated & 16.7 & $\{7,10\}$ \\
{}[[15,1,3]], bit & Prop.~\ref{prop:anyerr} & any & verified & 139 &  \\
{}[[15,1,3]], bit & Prop.~\ref{prop:detect} & 6 & verified & 49.3 &  \\
{}[[15,1,3]], bit & Prop.~\ref{prop:detect} & 7 & violated & 18.4 & $\{0,5,6,7,8,13,14\}$ \\
{}[[15,1,3]], bit & Prop.~\ref{prop:recover} & 1 & verified & 47.6 &  \\
{}[[15,1,3]], bit & Prop.~\ref{prop:recover} & 2 & violated & 16.5 & $\{2,10\}$ \\
{}[[15,1,3]], phase & Prop.~\ref{prop:anyerr} & any & verified & 193 &  \\
{}[[15,1,3]], phase & Prop.~\ref{prop:detect} & 2 & verified & 92.5 &  \\
{}[[15,1,3]], phase & Prop.~\ref{prop:detect} & 3 & violated & 16.7 & $\{0,7,8\}$ \\
{}[[15,1,3]], phase & Prop.~\ref{prop:recover} & 1 & verified & 112 &  \\
{}[[15,1,3]], phase & Prop.~\ref{prop:recover} & 2 & violated & 15.8 & $\{8,13\}$ \\
\bottomrule
\end{tabular}

\end{table}

\FloatBarrier
\subsection{Scalability for GHZ-state verification}\label{sec:resghz}
To evaluate scalability, we verified GHZ-state generation circuits ranging from 4 to 262,144 qubits using three topologies.
The \emph{star} topology used in~\cite{bib34} has all C-NOT gates controlled by the first qubit.
The \emph{chain} topology uses nearest-neighbour C-NOT gates with depth $m-1$, so the value of each qubit is propagated through all preceding qubits.
The \emph{tree} topology uses fan-out with logarithmic depth $\lceil\log_2 m\rceil$.

Selected results are given in Table~\ref{tab:ghz}, while Figure~\ref{fig:ghzscaling} summarises the overall scaling behaviour, including peak memory.
All instances are verified.
The largest star circuit, with 262,144 qubits, requires 9.01s of solve time, while the chain and tree circuits require 13.0s and 12.3s, respectively.

Across the evaluated range, solve time grows approximately linearly with circuit size
for all three topologies, consistent with the linear number of constraints discussed in Section~\ref{sec:ressize}.
The topology changes the solve time by less than 50\%, with the chain being the slowest and the star the fastest.

Compared with the largest GHZ-state experiment reported in~\cite{bib34}, 
which used 8,192 qubits, the present implementation verifies substantially larger circuits with fewer computational resources.
The ablation study in Section~\ref{sec:resablation} shows that this difference is not explained by the encoding pipeline or the solver version.
The comparison with~\cite{bib34} is therefore observational only, 
as discussed in Section~\ref{sec:disc-scal}.

For the largest instances, end-to-end time is dominated by Python-side construction of the constraint terms rather than by solving.
For $m=262,144$ in the star topology, construction takes 46.1s on average, compared with 9.01s of solve time.
Peak memory is also reached during term construction.
Measured in a fresh process for each instance, 
it grows linearly from the 47.0MiB baseline of the interpreter with Z3 loaded to 1,397MiB for 262,144 qubits in the star topology.
This corresponds to about 5.3KiB per qubit, which on the 16~GB machine would allow instances of about three million qubits before memory is exhausted.

\begin{table}[!h]
\centering
\caption{GHZ-state verification for three topologies (solve time: median $\pm$ s.d. over 10 repetitions, 5 for $m>16,384$; build: mean Python-side construction time; peak RSS: median of three fresh processes per instance, including the interpreter baseline).}
\label{tab:ghz}
\footnotesize
\begin{tabular}{@{}rrrrrrr@{}}
\toprule
 & \multicolumn{3}{c}{Solve time, median $\pm$ s.d. (s)} & Build (s) & \multicolumn{2}{c}{Peak RSS (MiB)} \\
\cmidrule(lr){2-4}\cmidrule(lr){6-7}
$m$ & star & chain & tree & star & star & tree \\
\midrule
4 & 0.0038 $\pm$ 0.0002 & 0.0032 $\pm$ 0.0003 & 0.0035 $\pm$ 0.0007 & 0.0007 & 50 & 50 \\
256 & 0.0102 $\pm$ 0.0002 & 0.0122 $\pm$ 0.0001 & 0.0139 $\pm$ 0.0004 & 0.0410 & 51 & 51 \\
4{,}096 & 0.0911 $\pm$ 0.0013 & 0.133 $\pm$ 0.0037 & 0.0985 $\pm$ 0.0021 & 0.675 & 60 & 61 \\
8{,}192 & 0.184 $\pm$ 0.0033 & 0.258 $\pm$ 0.0056 & 0.201 $\pm$ 0.0051 & 1.33 & 86 & 88 \\
16{,}384 & 0.364 $\pm$ 0.0062 & 0.468 $\pm$ 0.0169 & 0.415 $\pm$ 0.0061 & 2.69 & 137 & 142 \\
65{,}536 & 1.44 $\pm$ 0.0213 & 1.93 $\pm$ 0.0610 & 1.76 $\pm$ 0.0350 & 11.2 & 428 & 442 \\
131{,}072 & 3.14 $\pm$ 0.0805 & 4.58 $\pm$ 0.0555 & 3.70 $\pm$ 0.167 & 22.7 & 838 & 876 \\
262{,}144 & 9.01 $\pm$ 0.455 & 13.0 $\pm$ 0.667 & 12.3 $\pm$ 0.779 & 46.1 & 1,397 & 1,416 \\
\bottomrule
\end{tabular}

\end{table}

\begin{figure}[!h]
\centering
\includegraphics[width=\textwidth]{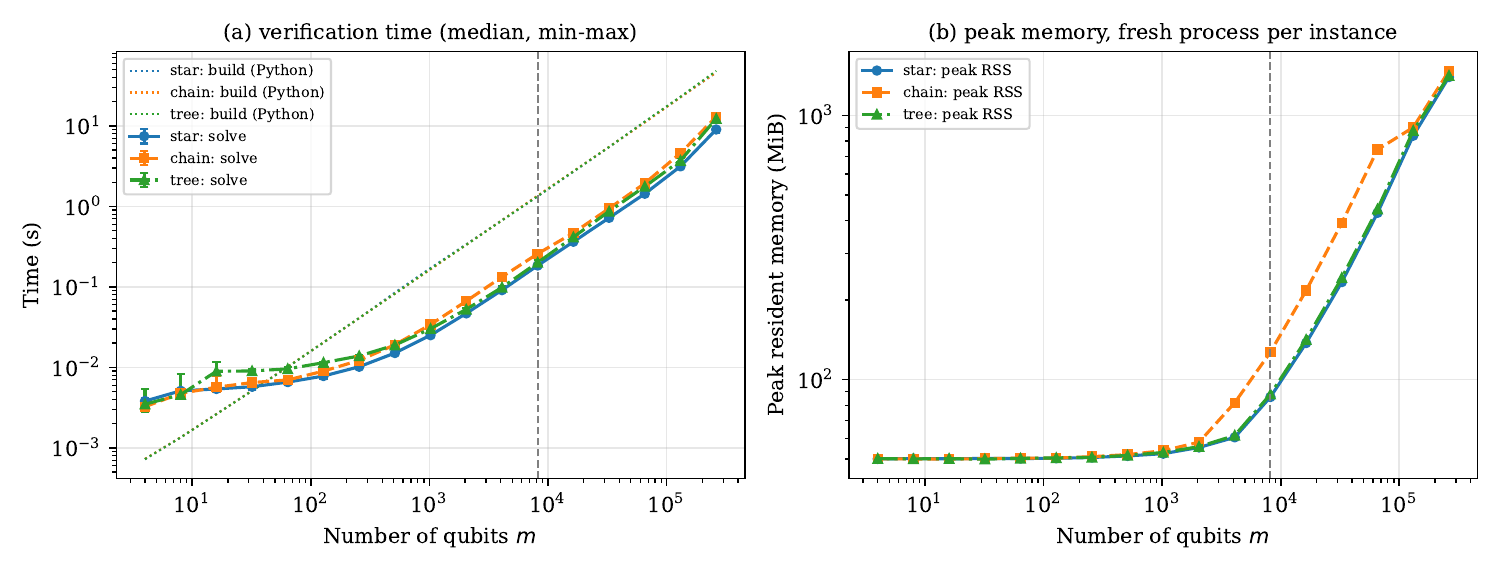}
\caption{GHZ-state verification versus qubit count for three topologies. (a) Solve time (median, minimum-maximum bars) and mean Python-side build time. (b) Peak resident memory of a fresh process per instance. The dashed line marks the largest GHZ circuit reported in~\cite{bib34}.}
\label{fig:ghzscaling}
\end{figure}

\FloatBarrier
\subsection{Encoding size}\label{sec:ressize}
Table~\ref{tab:size} reports the numbers of SMT variables, assertions and CNF clauses after bit-blasting for representative instances.
All three quantities grow linearly with the number of qubits, in agreement with the analysis in Section~\ref{sec:encoding}.
For GHZ circuits, the encoding produces about 6.5 clauses per qubit for the star topology and about 7.0 clauses per qubit for the chain and tree topologies.

\begin{table}[!h]
\centering
\caption{Size of the verification queries: SMT variables and assertions as generated, and CNF clauses after Z3's \texttt{simplify}, \texttt{card2bv}, \texttt{bit-blast} and \texttt{tseitin-cnf} tactics. Single-error properties for the code families, arbitrary-error property for the surface code (bit-flip frame).}
\label{tab:size}
\footnotesize
\begin{tabular}{@{}lrrrrr@{}}
\toprule
Family & Size & Qubits & SMT variables & Assertions & CNF clauses \\
\midrule
GHZ-star & 16 & 16 & 17 & 17 & 103 \\
GHZ-star & 64 & 64 & 65 & 65 & 419 \\
GHZ-star & 256 & 256 & 257 & 257 & 1{,}675 \\
GHZ-star & 1{,}024 & 1{,}024 & 1{,}025 & 1{,}025 & 6{,}669 \\
GHZ-star & 4{,}096 & 4{,}096 & 4{,}097 & 4{,}097 & 26{,}671 \\
GHZ-star & 16{,}384 & 16{,}384 & 16{,}385 & 16{,}385 & 106{,}574 \\
GHZ-chain & 16{,}384 & 16{,}384 & 16{,}385 & 16{,}385 & 114{,}684 \\
GHZ-tree & 16{,}384 & 16{,}384 & 16{,}385 & 16{,}385 & 114{,}678 \\
repetition & 3 & 5 & 13 & 11 & 83 \\
repetition & 9 & 17 & 43 & 35 & 320 \\
repetition & 21 & 41 & 103 & 83 & 844 \\
repetition & 51 & 101 & 253 & 203 & 2{,}183 \\
repetition & 101 & 201 & 503 & 403 & 4{,}470 \\
concatenated phase-flip & 3 & 5 & 19 & 17 & 197 \\
concatenated phase-flip & 9 & 17 & 61 & 53 & 662 \\
concatenated phase-flip & 21 & 41 & 145 & 125 & 1{,}642 \\
concatenated phase-flip & 51 & 101 & 355 & 305 & 4{,}121 \\
concatenated phase-flip & 101 & 201 & 705 & 605 & 8{,}291 \\
surface (bit, syndrome) & 3 & 13 & 43 & 30 & 271 \\
surface (bit, syndrome) & 5 & 37 & 127 & 90 & 863 \\
surface (bit, syndrome) & 9 & 121 & 427 & 306 & 3{,}043 \\
surface (bit, syndrome) & 15 & 337 & 1{,}207 & 870 & 8{,}803 \\
surface (bit, syndrome) & 21 & 661 & 2{,}383 & 1{,}722 & 17{,}551 \\
surface (bit, syndrome) & 31 & 1{,}441 & 5{,}223 & 3{,}782 & 38{,}771 \\
\bottomrule
\end{tabular}

\end{table}

\FloatBarrier
\subsection{Random H+C-NOT networks}\label{sec:resrandom}
GHZ, repetition and surface-code circuits produce sparse XOR relations, where each output depends on only a small number of variables.
To examine less favourable constraint structures,
we also verified random networks.
Each network applies Hadamard gates to a random half of the qubits as the first gate 
(branch qubits), followed by $D$ layers of C-NOT gates acting on random disjoint qubit pairs with random orientations, 
together with sparse X gates.

The property requires each output basis bit to equal the GF(2) 
affine function obtained independently from the gate list, 
and each tag to be $b'001$ exactly for the Hadamard qubits.
As $D$ increases, the number of input variables on which an output depends, 
referred to as its support, grows rapidly.

The results are shown in Table~\ref{tab:random} and Figure~\ref{fig:random}.
For shallow networks, verification remains fast even for thousands of qubits.
At a fixed width, however, solve time increases rapidly with depth.
For 32 qubits, the median solve time increases from 0.033s at depth~4 to 27s at depth~7, while at depth~8 two of the three seeds exceed the 300s timeout.
For 64 qubits, the median increases from 0.20s at depth~5 to 125s at depth~7.

\begin{table}[!h]
\centering
\caption{Random H+C-NOT networks of $n$ qubits and depth $D$ (three random seeds per configuration; 300\,s timeout). ``Max.\ support'' is the largest number of input variables on which an output depends.}
\label{tab:random}
\footnotesize
\begin{tabular}{@{}rrrrrl@{}}
\toprule
$n$ & Depth & C-NOTs & Max.\ support & Median time (s) & Result (3 seeds) \\
\midrule
16 & 1 & 8 & 2 & 0.0100 & 3$\times$ Yes \\
16 & 2 & 16 & 4 & 0.0201 & 3$\times$ Yes \\
16 & 4 & 32 & 10 & 0.0249 & 3$\times$ Yes \\
16 & 8 & 64 & 13 & 0.641 & 3$\times$ Yes \\
16 & 12 & 96 & 12 & 5.61 & 3$\times$ Yes \\
16 & 16 & 128 & 13 & 4.81 & 3$\times$ Yes \\
32 & 1 & 16 & 2 & 0.0120 & 3$\times$ Yes \\
32 & 2 & 32 & 4 & 0.0291 & 3$\times$ Yes \\
32 & 4 & 64 & 10 & 0.0329 & 3$\times$ Yes \\
32 & 6 & 96 & 17 & 1.38 & 3$\times$ Yes \\
32 & 7 & 112 & 19 & 27.5 & 3$\times$ Yes \\
32 & 8 & 128 & 19 & 300 & 2$\times$ timeout, 1$\times$ Yes \\
64 & 1 & 32 & 2 & 0.0106 & 3$\times$ Yes \\
64 & 2 & 64 & 4 & 0.0231 & 3$\times$ Yes \\
64 & 4 & 128 & 13 & 0.0417 & 3$\times$ Yes \\
64 & 5 & 160 & 16 & 0.205 & 3$\times$ Yes \\
64 & 6 & 192 & 22 & 4.61 & 3$\times$ Yes \\
64 & 7 & 224 & 32 & 125 & 3$\times$ Yes \\
256 & 1 & 128 & 2 & 0.0160 & 3$\times$ Yes \\
256 & 2 & 256 & 4 & 0.0392 & 3$\times$ Yes \\
256 & 3 & 384 & 8 & 0.0615 & 3$\times$ Yes \\
256 & 4 & 512 & 16 & 0.159 & 3$\times$ Yes \\
256 & 5 & 640 & 26 & 1.35 & 3$\times$ Yes \\
1{,}024 & 1 & 512 & 2 & 0.0443 & 3$\times$ Yes \\
1{,}024 & 2 & 1{,}024 & 4 & 0.103 & 3$\times$ Yes \\
1{,}024 & 3 & 1{,}536 & 8 & 0.294 & 3$\times$ Yes \\
4{,}096 & 1 & 2{,}048 & 2 & 0.170 & 3$\times$ Yes \\
4{,}096 & 2 & 4{,}096 & 4 & 0.458 & 3$\times$ Yes \\
\bottomrule
\end{tabular}

\end{table}

\begin{figure}[!h]
\centering
\includegraphics[width=0.8\textwidth]{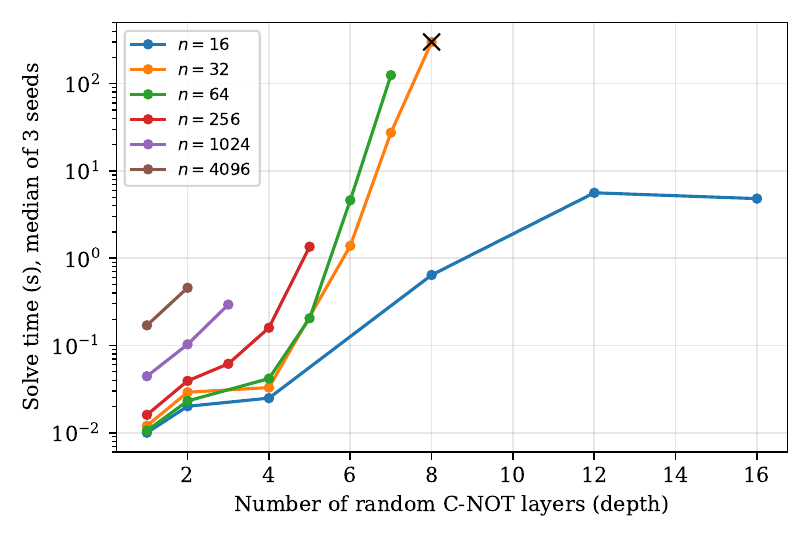}
\caption{Solve time of random H+C-NOT networks versus depth (median of three seeds; crosses mark configurations with at least one timeout).}
\label{fig:random}
\end{figure}

The circuit size is small in all of these cases.
The main difficulty is proving that two different XOR decompositions represent the same parity function, 
which is known to be difficult for resolution-based solvers~\cite{bib55}.
These results show that the favourable scaling observed for GHZ and code circuits 
depends on their sparse constraint structure rather than applying to the abstraction in general, 
as discussed in Section~\ref{sec:disc-scal}.

\FloatBarrier
\subsection{Controlled ablation of encoding pipeline and solver version}
\label{sec:resablation}
To identify the factors contributing to the difference with~\cite{bib34}, 
we compared the direct API pipeline with an SMT-LIB pipeline 
on the same machine and for identical constraint systems using GHZ star circuits.
In the SMT-LIB pipeline, the query is written to an SMT-LIB~2 file and then parsed into a fresh solver before solving.
The same file is also solved using the stand-alone \texttt{z3} executable of the same release in a separate process.
We also compared two Z3 releases.

Table~\ref{tab:ablation} shows that, for identical constraints, 
the SMT-LIB pipeline is \emph{not} slower than the direct API.
Serialisation and parsing add a cost that grows with query size, 
with 4.1s and 1.9s for 262,144 qubits and a 50MB file.
However, the stand-alone \texttt{z3} executable solves the SMT-LIB file, 
including start-up and parsing, in 5.0s, whereas solving the API-built query takes 8.8s.
The same ordering is observed at all tested sizes.

The measured difference between the two solve times appears to result 
from differences in the term structure of the two inputs: 
the API query contains one large conjunction, whereas the parsed file contains individual assertions.
This was not investigated further.
Thus, the direct API encoding does \emph{not} explain the improvement over~\cite{bib34}, 
and is treated as an observation only.
Its practical advantage is instead that it avoids intermediate files and supports the programmatic construction of circuit families, mutants and miters, rather than providing faster solving.

\begin{table}[!h]
\centering
\caption{Direct API versus SMT-LIB pipeline for identical GHZ-star queries (median of three repetitions).}
\label{tab:ablation}
\footnotesize
\begin{tabular}{@{}lrrrrrrr@{}}
\toprule
 & & Direct API & \multicolumn{4}{c}{SMT-LIB pipeline, in process (s)} & External \\
\cmidrule(lr){4-7}
Z3 & $m$ & solve (s) & serialise & parse & solve & total & \texttt{z3} (s) \\
\midrule
4.13.4 & 1{,}024 & 0.0174 & 0.0155 & 0.0049 & 0.0035 & 0.0242 & 0.0176 \\
4.13.4 & 4{,}096 & 0.0668 & 0.0608 & 0.0204 & 0.0140 & 0.0965 & 0.0433 \\
4.13.4 & 16{,}384 & 0.272 & 0.245 & 0.0949 & 0.0688 & 0.410 & 0.171 \\
4.13.4 & 65{,}536 & 1.19 & 1.00 & 0.435 & 0.377 & 1.84 & 0.769 \\
4.13.4 & 262{,}144 & 8.80 & 4.11 & 1.93 & 5.45 & 11.4 & 5.05 \\
5.1.0 & 1{,}024 & 0.0181 & 0.0159 & 0.0052 & 0.0037 & 0.0252 & 0.0182 \\
5.1.0 & 4{,}096 & 0.0688 & 0.0620 & 0.0206 & 0.0147 & 0.0982 & 0.0452 \\
5.1.0 & 16{,}384 & 0.281 & 0.249 & 0.0958 & 0.0628 & 0.413 & 0.175 \\
5.1.0 & 65{,}536 & 1.19 & 1.02 & 0.424 & 0.349 & 1.79 & 0.779 \\
5.1.0 & 262{,}144 & 8.47 & 4.19 & 1.87 & 5.01 & 11.1 & 5.01 \\
\bottomrule
\end{tabular}

\end{table}

\begin{table}[!h]
\centering
\caption{Solve time of GHZ-star queries for two Z3 releases on the same machine (median $\pm$ s.d.).}
\label{tab:version}
\footnotesize
\begin{tabular}{@{}rrr@{}}
\toprule
$m$ & Z3 4.13.4 (s) & Z3 5.1.0 (s) \\
\midrule
4{,}096 & 0.0911 $\pm$ 0.0013 & 0.0700 $\pm$ 0.0010 \\
16{,}384 & 0.364 $\pm$ 0.0062 & 0.294 $\pm$ 0.0058 \\
65{,}536 & 1.44 $\pm$ 0.0213 & 1.22 $\pm$ 0.0221 \\
131{,}072 & 3.14 $\pm$ 0.0805 & 3.13 $\pm$ 0.0950 \\
262{,}144 & 9.01 $\pm$ 0.455 & 9.01 $\pm$ 0.199 \\
\bottomrule
\end{tabular}

\end{table}

Table~\ref{tab:version} compares the solve times of Z3~4.13.4 and Z3~5.1.0.
The two versions are within about 25\% of each other at all tested sizes and give identical times at the largest size.
The solver version is therefore not a significant factor in the observed difference.

These ablation experiments isolate the encoding pipeline and solver version on the same hardware.
They do not reproduce the implementation used in~\cite{bib34}, 
since its SMT-LIB generator, the precise form of the generated constraints and the hardware all differ.
The source of the difference from~\cite{bib34} therefore remains unresolved, and the comparison is purely observational.

\FloatBarrier
\subsection{Differential validation of the soundness theorem}
\label{sec:ressound}
To provide an independent check of Theorem~\ref{thm:soundness}, we compared the abstraction with exact state-vector simulation.
For random circuits in $\mathcal{F}$ on 2-5 qubits, we used random qubit roles, 
random X and C-NOT gates, and random frame segments containing random phase-flip errors.
The deep set contained up to eight frame segments per qubit.
For every computational-basis input, we checked Eq.~\eqref{eq:soundness}.
Specifically, we verified that the support of the output state matched the set of basis states
predicted by the abstraction over all branch assignments, 
that all amplitudes had magnitude $2^{-|B|/2}$, and that the odd tags occurred exactly on the branch qubits.

\begin{table}[!h]
\centering
\caption{Differential test of the abstraction against state-vector simulation.}
\label{tab:sound}
\footnotesize
\begin{tabular}{@{}llrr@{}}
\toprule
Random circuits & $n$ & Circuits & Agree on all inputs \\
\midrule
In $\mathcal{F}$ & 2 & 2000 & 2000 \\
In $\mathcal{F}$ & 3 & 2000 & 2000 \\
In $\mathcal{F}$ & 4 & 2000 & 2000 \\
In $\mathcal{F}$ & 5 & 2000 & 2000 \\
In $\mathcal{F}$, deep (639 saturate) & 2 & 2000 & 2000 \\
In $\mathcal{F}$, deep (889 saturate) & 3 & 2000 & 2000 \\
Unrestricted & 2 & 2000 & 1108 \\
Unrestricted & 3 & 2000 & 1073 \\
Unrestricted & 4 & 2000 & 1120 \\
\bottomrule
\end{tabular}

\end{table}

Table~\ref{tab:sound} shows that all 12,000 circuits in $\mathcal{F}$ agreed for all inputs, including 1,528 circuits in which the tag saturated.
For unrestricted random circuits over $\{H,X,\text{C-NOT}\}$, 
only about 55\% agreed, and each of the counterexamples C1-C4 from Remark~\ref{rem:cex} produced a disagreement.
These results support the soundness theorem on its stated fragment and show the need for the restriction to $\mathcal{F}$.

\FloatBarrier
\subsection{Mutation-based validation}\label{sec:resmutation}
The proposed correctness properties were also evaluated using the mutation-testing framework described in Section~\ref{sec:mutationmethod}.
Applying every operator in Table~\ref{tab:mutops} at every valid location of every instance generated 682 mutants across 12 circuit families.
A further 12 mutations had no valid location and were therefore not generated.

The SMT miter identified 60 of the generated mutants as equivalent.
These were exactly the \texttt{reorder\_cnot} mutants, 
which reverse the order of commuting C-NOT gates, and no other mutants.
They were excluded from the fault-detection analysis.
\begin{table}[!h]
\centering
\caption{Mutation results per circuit family.}
\label{tab:mutfam}
\footnotesize
\begin{tabular}{@{}lrrrrrr@{}}
\toprule
Circuit family & Generated & Not gen. & Equivalent & Evaluated & Detected & Rate \\
\midrule
repetition & 140 & 0 & 20 & 120 & 120 & 100.0\% \\
concatenated phase-flip & 129 & 0 & 12 & 117 & 117 & 100.0\% \\
Steane bit-flip & 21 & 0 & 3 & 18 & 18 & 100.0\% \\
Steane phase-flip & 42 & 0 & 3 & 39 & 39 & 100.0\% \\
surface d=3 bit-flip & 28 & 0 & 4 & 24 & 24 & 100.0\% \\
surface d=3 phase-flip & 55 & 0 & 4 & 51 & 51 & 100.0\% \\
Reed-Muller [[15,1,3]] bit-flip & 70 & 0 & 10 & 60 & 60 & 100.0\% \\
Reed-Muller [[15,1,3]] phase-flip & 73 & 0 & 4 & 69 & 69 & 100.0\% \\
Bell & 4 & 0 & 0 & 4 & 4 & 100.0\% \\
GHZ star & 40 & 4 & 0 & 40 & 40 & 100.0\% \\
GHZ chain & 40 & 4 & 0 & 40 & 40 & 100.0\% \\
GHZ tree & 40 & 4 & 0 & 40 & 40 & 100.0\% \\
\midrule
\textbf{Total} & 682 & 12 & 60 & 622 & 622 & 100.0\% \\
\bottomrule
\end{tabular}

\end{table}

\begin{table}[!h]
\centering
\caption{Mutation results per operator (all families).}
\label{tab:mutop}
\footnotesize
\begin{tabular}{@{}lrrrrrr@{}}
\toprule
Operator & Generated & Not gen. & Equivalent & Evaluated & Detected & Rate \\
\midrule
\texttt{drop\_cnot} & 60 & 0 & 0 & 60 & 60 & 100.0\% \\
\texttt{wrong\_wire} & 60 & 0 & 0 & 60 & 60 & 100.0\% \\
\texttt{extra\_cnot} & 60 & 0 & 0 & 60 & 60 & 100.0\% \\
\texttt{swapped\_cnot} & 97 & 0 & 0 & 97 & 97 & 100.0\% \\
\texttt{reorder\_cnot} & 60 & 0 & 60 & 0 & 0 & -- \\
\texttt{mis\_init} & 60 & 0 & 0 & 60 & 60 & 100.0\% \\
\texttt{extra\_h\_anc} & 60 & 0 & 0 & 60 & 60 & 100.0\% \\
\texttt{missing\_h1} & 46 & 0 & 0 & 46 & 46 & 100.0\% \\
\texttt{missing\_h2} & 46 & 0 & 0 & 46 & 46 & 100.0\% \\
\texttt{extra\_h} & 59 & 0 & 0 & 59 & 59 & 100.0\% \\
\texttt{missing\_h} & 13 & 0 & 0 & 13 & 13 & 100.0\% \\
\texttt{missing\_cnot} & 37 & 0 & 0 & 37 & 37 & 100.0\% \\
\texttt{wrong\_target} & 24 & 12 & 0 & 24 & 24 & 100.0\% \\
\midrule
\textbf{Total} & 682 & 12 & 60 & 622 & 622 & 100.0\% \\
\bottomrule
\end{tabular}

\end{table}

\begin{figure}[!h]
\centering
\includegraphics[width=0.9\textwidth]{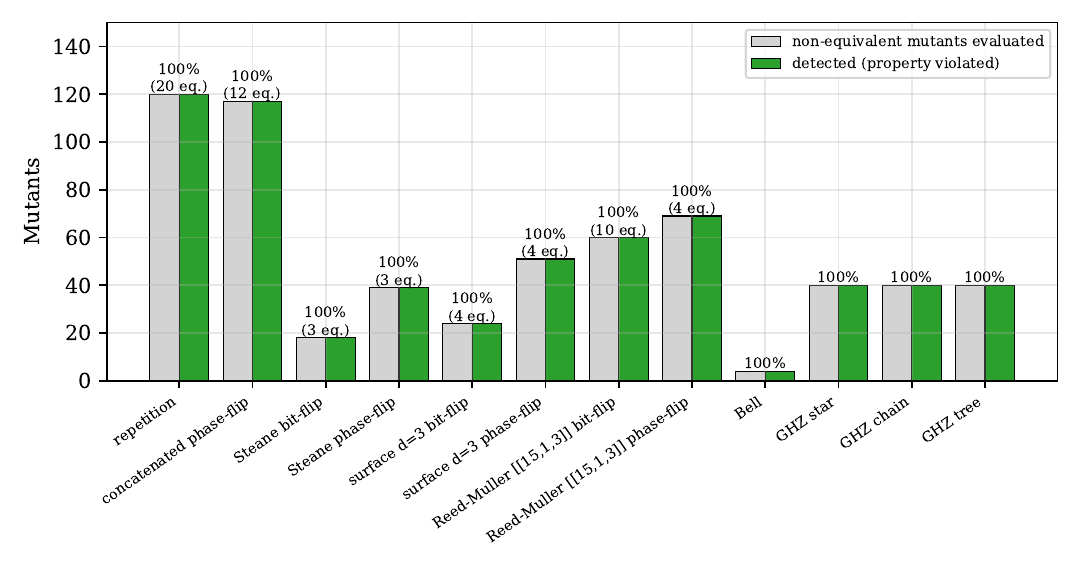}
\caption{Non-equivalent mutants evaluated and detected per circuit family.}
\label{fig:mutation}
\end{figure}

Tables~\ref{tab:mutfam} and~\ref{tab:mutop} summarise the results for the remaining 622 non-equivalent mutants, 
and Figure~\ref{fig:mutation} shows the detection rates.
Every non-equivalent mutant produced a property violation that was detected by the SMT solver, giving a 100\% detection rate for every circuit family and every operator.

These results are consistent with Lemmas~\ref{lem:1}-\ref{lem:6} 
and provide empirical evidence that the \emph{implemented} 
properties detect the implementation fault classes considered in this study.
The 100\% detection rate applies only to the evaluated mutation 
classes and circuit instances and does not imply complete fault coverage (Section~\ref{sec:discussion}).

\FloatBarrier
\subsection{Discussion of results}
The experiments show that the proposed abstraction verifies all benchmark circuits 
considered in this study while remaining computationally efficient 
for circuits with sparse XOR structure.
The generalised repetition and concatenated phase-flip experiments show that the correctness properties extend to parametrised circuit families.
The CSS experiments show that the abstraction can be used for syndrome extraction with arbitrary CSS check matrices.
The multiple-error experiments distinguish detection, 
identification and recovery and confirm the expected limits of the codes considered.

The GHZ experiments show that the implementation scales to circuits with more than 
$2.6\times10^5$ qubits for all three evaluated topologies under the proposed abstraction.
The random-network and recovery experiments also show the limits of this scaling.
The differential test provides an independent check of the soundness theorem on its stated fragment and shows its failure outside that fragment.
The mutation study provides empirical evidence that the implemented correctness properties detect the fault classes considered in Section~\ref{sec:proofs}.

Taken together, the results show that the proposed verification framework can handle the circuit classes studied here at substantial scale,
while also identifying constraint structures and decoder properties that lead to higher verification costs.

\section{Discussion: Scope, Limitations, and Threats to Validity}
\label{sec:discussion}

This section discusses the scope of the proposed abstraction, 
its current limitations, and the main threats to the validity of the experimental evaluation.

\subsection{Scope of the abstraction}\label{sec:scope}
The abstraction introduced in Section~\ref{sec:absfrmwrk}, together with the soundness result in Theorem~\ref{thm:soundness},
applies to circuits in the fragment $\mathcal{F}$ defined in Definition~\ref{def:fragment}. 
This fragment is a proper subset of the circuits composed of $H$, $X$ and C-NOT gates, 
and therefore the method is not a verification technique for general stabilizer circuits.

The fragment is sufficient for all circuit families studied in this work, 
including repetition and concatenated codes, Bell and GHZ state preparation, random H+C-NOT networks, and CSS syndrome extraction. 
This includes CSS codes with $\mathbf{H}_X\neq\mathbf{H}_Z$, such as the $[[15,1,3]]$ 
quantum Reed-Muller code and the surface codes. 
Within these circuits, phase-flip errors are represented using the frame interpretation described in Section~\ref{sec:baseabs}, 
rather than by an explicit $Z$ gate. 
A physical bit-flip fault occurring \emph{inside} a Hadamard frame segment is therefore outside the model. 
As shown by counterexample C1 in Remark~\ref{rem:cex}, the abstraction would interpret such a fault as a phase flip.

The abstraction also does not directly cover stabilizer codes whose generators contain intrinsic $Y$-type Pauli operators 
that cannot be represented as purely $X$-type or purely $Z$-type operators on the same support. 
A representative example is the five-qubit perfect code~\cite{bib46}, 
whose generators $XZZXI$ and their cyclic shifts contain both $X$ and $Z$ 
factors on overlapping qubit sets, so that products of generators can contain $Y$ factors. 
Its stabilizer group therefore cannot be generated entirely from separate $X$-type and $Z$-type elements and has no CSS decomposition~\cite{bib47}. 
Supporting such codes would require additional gate semantics, in particular an abstract phase ($S$) gate, together with a representation of the relative phases $\omega$ in Eq.~\eqref{eq:soundness}, 
rather than only the current superposition status (Section~\ref{sec8}).

\paragraph{Circuit-level verification versus physical validation.}
The framework verifies logical properties of an abstract circuit description. 
These include which syndrome is computed, whether gates are connected as intended, and whether the expected entanglement structure is produced.
It does not model hardware noise, decoherence, leakage, measurement errors, gate infidelities, crosstalk, timing effects, or other physical-device behaviour.
In this work, the terms \emph{fault} and \emph{mutant} refer to design or implementation errors in the circuit description, 
whereas \emph{errors} denote the discrete Pauli bit-flip and phase-flip patterns that the code is intended to handle.
Establishing that a physical device correctly implements the verified circuit, 
or determining its logical error rate under a realistic noise model, 
requires complementary methods such as noise-aware simulation~\cite{bib41} and experimental characterisation.

\subsection{Detection, correction, and fault tolerance}\label{sec:faulttol}
Section~\ref{sec:detcorr} distinguishes detection, identification, decoding, 
recovery and fault tolerance.
The properties of~\cite{bib34} and Properties~\ref{prop:genrep}-\ref{prop:css} verify correct syndrome extraction and error identification under a single error.
Property~\ref{prop:anyerr} extends syndrome verification to arbitrary error patterns, 
while Property~\ref{prop:detect} verifies detection up to a specified error weight.
Property~\ref{prop:recover} additionally verifies that a specified decoder, 
together with classically controlled corrections, returns the data to the error-free code word for all errors up to weight $w$.
For the generalised repetition code with $n>3$, the single-error properties establish identification only. 
Correction of multiple errors is established by Property~\ref{prop:recover} for the values of $n$ and $w$ that were verified in Section~\ref{sec:resmulti}.

These properties do not establish \emph{fault tolerance}.
In all verified circuits, the error variables act on the data qubits \emph{before} 
syndrome extraction, while the extraction circuit itself is assumed to be fault-free.
Faults affecting ancillas or C-NOT gates during extraction are not modelled, 
and neither is their propagation.

In particular, a phase fault on an ancilla that is the target of several C-NOT gates can propagate to the control (data) qubits through phase kickback.
Similarly, an ancilla prepared in $|+\rangle$ that controls several 
C-NOT gates can propagate an $X$ fault to subsequent targets.
Either mechanism can convert a single fault into a correlated multi-qubit data error, 
commonly referred to as a hook error~\cite{bib49}.
The present abstraction cannot represent these mechanisms.
Phase kickback is a phase effect that is not captured by Definition~\ref{def:cnot}.
The second mechanism requires a qubit that has acted as a C-NOT control while in superposition to receive a second Hadamard gate, 
which is excluded from $\mathcal{F}$, as shown by counterexample C2.

Consequently, for the Steane code, the framework verifies the classical Hamming$(7,4)$ 
syndrome logic of the ancilla-based extraction circuit in both Pauli frames.
It does not verify the fault tolerance of the complete extraction circuit.
Such verification would require, for example, modelling verified ancilla preparation, 
repeated extraction or flag qubits~\cite{bib6,bib48}.
More generally, fault-tolerance verification would require error variables at each location 
in the extraction circuit and an error-propagation semantics 
that includes phase kickback, as outlined in the phase-aware extension in Section~\ref{sec8}.

\subsection{Error model and preconditions}\label{sec:errmodel}
The single-error properties assume that at most one error is present.
For the Steane and CSS properties, 
they also assume that the input is a valid code word, 
namely $\mathbf{H}_Zv=0$ or $\mathbf{H}_Xv=0$, respectively.
This code-word precondition is part of the specification of syndrome extraction: 
the syndrome of a non-code word is not an error syndrome.
Producing a valid code word is therefore the responsibility of the encoding circuit, 
which can be verified separately and then composed with the extraction circuit.

The single-error assumption, in contrast, is not intrinsic to syndrome extraction.
Property~\ref{prop:anyerr} removes this restriction and verifies syndrome extraction for every error pattern.
Properties~\ref{prop:detect} and~\ref{prop:recover} then specify the error weights up to which detection and recovery are required to hold.
The current model does not include errors that occur during syndrome extraction or recovery (Section~\ref{sec:faulttol}), 
measurement errors, or repeated rounds of syndrome extraction.

\subsection{Coverage of the mutation analysis}
The mutation study in Section~\ref{sec:resmutation} focuses on implementation faults directly covered by Lemmas~\ref{lem:1}-\ref{lem:6}.
These include incorrect ancilla initialisation, omitted gates, 
additional gates and incorrect C-NOT wiring.
Each mutant contains a single injected modification.
Although combinations of modifications are covered by the characterisation in Lemma~\ref{lem:6}, they were not generated systematically in the mutation experiment.

Other fault classes were not considered.
These include the insertion of gates outside the supported $\{H,X,\mathrm{C\mbox{-}NOT}\}$ set,
such as $Z$ or $S$ gates, timing-related faults in physical implementations, 
and systematically generated combinations of multiple simultaneous faults.

The reported 100\% detection rate should therefore be interpreted narrowly.
It means that every non-equivalent single-modification mutant from the listed operator classes, 
for the listed circuit instances, violated its corresponding property.
It is not evidence of complete fault coverage or of detection of physically realistic fault classes.
Its main purpose is to show that the implemented encodings behave as predicted by Lemmas~\ref{lem:1}-\ref{lem:6}, as discussed in Section~\ref{sec:mutationmethod}.

\subsection{Interpretation of the scalability results}\label{sec:disc-scal}
The GHZ-state experiments show substantially improved verification performance compared with the implementation reported in~\cite{bib34}.
This comparison requires some care for three reasons.

First, the GHZ, repetition, concatenated and surface-code queries contain sparse linear XOR relationships 
in which each output depends on only a small number of variables.
These structures are well suited to SMT solving, 
and their solve times grow approximately linearly over the measured range.
The reported scale is therefore a property of these circuit families under the proposed abstraction.
It should not be interpreted as evidence that arbitrary circuits with $2.6\times10^5$ qubits can be verified.

The random H+C-NOT experiments in Section~\ref{sec:resrandom} show the other side of this result.
When outputs depend on many input variables, 
proving equality between the XOR structure of the circuit 
and that of the specification becomes substantially harder.
At fixed width, solve times therefore increase rapidly with circuit depth.
This behaviour is consistent with the known difficulty of parity reasoning for resolution-based solvers~\cite{bib55}.
Solvers with native XOR reasoning~\cite{bib54}, or approaches based on Gaussian elimination, 
would be natural alternatives for these constraints because the underlying parity problems are decidable in polynomial time by linear algebra.

Second, the reported scale should be considered in the context of stabilizer simulation~\cite{bib40,bib41}.
GHZ preparation is a Clifford circuit and can be simulated in polynomial 
time using tableau-based methods.
The SMT queries address a different task: they verify a property for all admissible inputs and error patterns simultaneously, or return a counterexample.
The results therefore measure the ability of the solver to handle the 
linear bit-vector constraints produced by the abstraction, 
rather than general symbolic reasoning over high-dimensional Hilbert spaces.

Third, the comparison with~\cite{bib34} is not a controlled experiment because the original implementation was not re-executed under identical conditions and the hardware differs.
Section~\ref{sec:resablation} therefore isolates the factors that could be controlled.
On the same hardware and for identical constraints, 
it compares the direct API pipeline with an SMT-LIB serialisation pipeline 
and also compares two Z3 releases.
Neither factor explains the observed improvement.
In particular, the stand-alone solver is faster on the SMT-LIB file than solving the API-built query, and the two Z3 releases have similar performance.
The remaining differences, including the hardware and the details of the constraint generator used in~\cite{bib34}, could not be isolated.
The comparison with~\cite{bib34} should consequently be regarded as observational rather than controlled.

Finally, beyond 262,144 qubits, the main resource limitation is Python-side construction and memory rather than solver time (Section~\ref{sec:resghz}).
At the largest size, construction of the Python term graph takes about five times longer than solving, 46s compared with 9s.
Peak memory is also reached during construction and grows approximately linearly at around 5.3KiB per qubit.
On the 16~GB machine used in the experiments, this corresponds to a limit of approximately three million qubits before the available memory is exhausted.

\subsection{Threats to validity of the soundness argument}
Theorem~\ref{thm:soundness} is proved for the fragment $\mathcal{F}$, 
and every benchmark circuit used in the experiments belongs to this fragment by construction.
The theorem does not impose a bound on the number of Hadamard gates applied to a qubit.
A frame qubit may pass through any number of frame segments, 
corresponding to repeated encode-decode cycles, because the saturating tag preserves parity (Definition~\ref{def:h}).

The theorem does require that, between two Hadamard gates, 
a qubit participates in no gate other than phase-flip error injections.
It also requires that a qubit which has interacted with another qubit does not receive a further Hadamard gate.
The need for these conditions is illustrated by counterexamples C1-C4.
The differential test in Section~\ref{sec:ressound} further supports the theorem, 
including tests with random circuits containing up to eight frame segments per qubit and therefore exercising tag saturation.

Circuits outside $\mathcal{F}$, including circuits that apply Hadamard gates to entangled qubits,
circuits containing $S$ or $Z$ gates, 
or circuits in which data qubits in superposition act as C-NOT controls between Hadamard gates,
require a richer abstract state.
These cases are left for future work.

\subsection{Applicability to CSS codes}\label{sec:disc-css}
The Steane experiment demonstrated the feasibility of the approach 
for one CSS construction in which the same Hamming matrix defines both check types.
The evaluation extends this to a code with $\mathbf{H}_X=\mathbf{H}_Z$ and higher check weight,
namely the $[[15,7,3]]$ code; a code with $\mathbf{H}_X\neq\mathbf{H}_Z$ and asymmetric distances,
namely the $[[15,1,3]]$ quantum Reed-Muller code with $d_X=7$ and $d_Z=3$; 
and a family with a two-dimensional local check structure and degenerate syndromes, 
namely the rotated surface codes.

All of these cases are handled by Definition~\ref{def:css} without modifying the underlying gate abstraction.
This supports the use of the framework for CSS syndrome extraction because the corresponding property is defined by the pair $(\mathbf{H}_X,\mathbf{H}_Z)$.
The scope of this result is limited to syndrome extraction, detection, and look-up or minimum-weight recovery.
It does not extend to fault-tolerant implementations or to decoders for degenerate codes, 
such as minimum-weight matching for surface codes, whose verification is left for future work.

\section{Conclusion and Future Directions}\label{sec8}

This paper extends the bit-vector abstraction proposed in~\cite{bib34} for the formal verification of quantum error-detection and entanglement circuits in several directions. 
First, the original correctness properties are generalised from fixed-size examples to parametrised repetition and concatenated code families.
Second, the abstraction is applied to the Steane $[[7,1,3]]$ CSS code, demonstrating that the framework is applicable to a representative quantum error-correcting code beyond repetition-based constructions.
Third, a soundness theorem is established for the abstraction within its supported gate set, and complete proofs are provided for the six lemmas that were stated without proof in the original work. 
Fourth, these theoretical results are supported by a mutation-based evaluation comprising 120 circuit mutants, all of which were successfully detected by the proposed verification properties. 
Finally, the verification framework is reimplemented using a direct Z3 API encoding, 
enabling GHZ-state verification at substantially larger scales than previously reported while using considerably fewer computational resources.

Taken together, these results indicate that the proposed abstraction is not limited to the fixed-size examples considered in the original study.
The generalised correctness properties, the verification of the Steane code, and the accompanying theoretical and empirical analyses demonstrate that the framework can be applied to a broader class of quantum circuits while maintaining efficient SMT-based verification.

Several directions remain for future research. 
The first is to extend the abstraction beyond the current $\{H,X,\mathrm{C\mbox{-}NOT}\}$ gate set 
to support stabilizer codes containing genuine $Y$-type operators, such as the five-qubit perfect code. 
This will require enriching the abstract state representation to capture relative phase information. 
A second direction is to move beyond the current single-error assumption by modelling multiple simultaneous errors together with the corresponding recovery operations, thereby enabling end-to-end verification of complete fault-tolerant protocols. 
Finally, the mutation-based evaluation framework presented in this work could be extended and adopted more broadly as a complementary validation methodology for quantum circuit verification studies, providing empirical evidence for the fault-detection capability of proposed correctness properties.

Overall, the experimental results demonstrate that the proposed SMT-based verification framework remains computationally efficient across the circuit families investigated, 
while the theoretical analysis establishes the correctness guarantees of the abstraction within its intended scope. 
These findings further support the use of SMT solving as a practical approach for the formal verification of quantum circuits constructed from the supported gate set.

\bibliographystyle{unsrtnat}
\bibliography{ref}

\end{document}